\documentclass{amsart}
\usepackage{amsmath}
\usepackage{fancyhdr}
\usepackage{latexsym}
\usepackage{amssymb}

\theoremstyle{plain}
\newtheorem{theorem}{Theorem}[section]
\newtheorem{lemma}[theorem]{Lemma}

\newtheorem{corollary}[theorem]{Corollary}
\newtheorem{assumption}[theorem]{Standing Assumption}
\newtheorem{assumptions}[theorem]{Standing Assumptions}

\theoremstyle{definition}
\newtheorem{definition}[theorem]{Definition}

\newtheorem{remarks}[theorem]{Remarks}
\newtheorem{example}[theorem]{Example}

\newcommand{\Nat}{{\mathbb N}}

\newcommand\GEX{{\Gamma\!\sb{\rm ex}}}
\newcommand{\subsim}{{\,\stackrel{\sb{<}}{\sb{\sim}}\,}}
\newcommand{\etasim}{{\sim\sb{\eta}}}
\title {Dimension theory for generalized effect algebras}
\author{David J. Foulis and Sylvia Pulmannov\'a}

\thanks{The second author was supported by the Center of Excellence
SAS -~Quantum Technologies; ERDF OP R\&D Project  meta-QUTE ITMS 26240120022; 
the grant VEGA No. 2/0059/12 SAV; the Slovak
Research and Development Agency under the contract APVV-0178-11}

\begin{document}

\address{Department of Mathematics and Statistics, University of Massachusetts,
Amherst, MA, USA; Mathematical Institute, Slovak Academy of Sciences, \v
Stef\'anikova 49, 814 73 Bratislava, Slovakia}

\email{foulis@math.umass.edu, pulmann@mat.savba.sk}
\keywords{generalized effect algebra, effect algebra, Dedekind orthocomplete
direct summand, center, exocenter, centrally orthocomplete, hull system, hull
determining, $\eta$-exocenter, $\eta$-monad, $\eta$-dyad, SK-congruence,
hereditary, $\pi$ splits $\sim$, invariant element, dimension equivalence
relation, dimension generalized effect algebra, simple element, finite
element, faithful element, type I/II/III decomposition.}
\subjclass{Primary 08A55, 81P10, Secondary 03G12}
\maketitle
\markboth{Foulis, D. and Pulmannov\'a, S.}{Dimension theory for GEAs}
\date{}
\begin{abstract}
In this paper we define and study dimension generalized effect algebras
(DGEAs), i.e., Dedekind orthocomplete and centrally orthocomplete
generalized effect algebras equipped with a dimension equivalence relation.
Our theory is a \emph{bona fide} generalization of the theory of dimension
effect algebras (DEAs), i.e., it is formulated so that, if a DGEA happens
to be an effect algebra (i.e., it has a unit element), then it is a DEA.
We prove that a DGEA decomposes into type I, II, and III DGEAS in a
manner analogous to the type I/II/III decomposition of a DEA.
\end{abstract}

\maketitle

\section{Introduction} \label{sc:Intro} 

\noindent In \cite{HandD}, we generalized to effect algebras the
Loomis-Maeda dimension theory for orthomodular lattices \cite{Loom,
Mae}; in the present paper, we extend the dimension theory to
generalized effect algebras (GEAs). Our theory complements
the dimension theory of espaliers developed by K. Goodearl and
F. Wehrung \cite{GWDim} inasmuch as, like an espalier, a GEA is
partially ordered, need not have a largest element, and hosts
a partially defined orthogonal summation. Unlike a GEA, every
nonempty subset of an espalier is required to have an infimum,
and if two orthogonal elements of an espalier have an upper bound,
then their orthogonal sum is their supremum. Unlike an espalier, a GEA
is required to satisfy a cancellation law for the orthogonal
summation.

This paper should be regarded as a continuation of \cite{ExoCen, CenGEA,
HDTD}. Although we shall briefly review some of the more pertinent
definitions, the reader is encouraged to consult the cited papers for
details of the basic theory of GEAs.

A recent surge of interest in GEAs can be partially attributed to the
observation that various systems of (possibly) unbounded symmetric
operators on a Hilbert space, in particular, operators that represent
quantum observables and states, can be organized into GEAs \cite{Pola,
PZOps, PVQObs, RZP}.

Generalized effect algebras are a subclass of the class of partial abelian
semigroups \cite{AlexPAS} or partial abelian monoids \cite{JePu02}. In fact,
a GEA is the same thing as a positive and cancellative partial abelian monoid
(pc-PAM) \cite[Definition 2.2]{JePu02}. Moreover, GEAs coincide with
generalized difference posets \cite{HedPu} and also with abelian RI-sets
\cite{KalmZdenka}. Numerous examples of GEAs can be found in the cited
references.

In developing our theory of \emph{dimension GEAs} (DGEAs)
we shall keep in mind two desiderata:
\begin{enumerate}
\item[{(1)}] If a DGEA is an effect algebra (EA), i.e., if it has a unit
 element, then it is a dimension effect algebra (DEA) in the sense of
 \cite[\S 7]{HandD}.
\item[{(2)}] Most of the important features of a DEA, including the type
 I/II/III decomposition \cite[Theorem 11.3]{HandD}, should go through for
 a DGEA.
\end{enumerate}

Direct decompositions of an effect algebra (EA) correspond to elements of
the center of the EA, but this is no longer the case for the center of a GEA
\cite[Definition 4.1]{ExoCen}, \cite[Definition 4.6]{CenGEA}; hence, for a
GEA, the center is superseded by the so-called \emph{exocenter}
\cite[Definition 3.1]{ExoCen}. Likewise, for GEAs, the notion of a hull
mapping on an EA \cite[Definition 3.1]{HandD}, which plays an important
role in the dimension theory of EAs, has to be replaced by a so-called
\emph{hull system} \cite[Definition 7.1]{ExoCen}.

The dimension theory of an EA employs the assumption of \emph{orthocompleteness}
(i.e., every orthogonal family is orthosummable) \cite[\S 7]{HandD}, but this
condition is too strong to impose on a GEA---in fact, an orthocomplete GEA that
is upward directed is an EA \cite[Theorem 5.5]{ExoCen}. Thus, in formulating our
definition of a \emph{dimension GEA} (DGEA), we assume only the weaker condition
of \emph{Dedekind orthocompleteness} (i.e., every orthogonal family with bounded
finite partial orthosums is orthosummable) \cite[Definition 5.1]{ExoCen}.
Obviously, an EA is orthocomplete iff it is Dedekind orthocomplete. An
orthocomplete EA is automatically \emph{centrally orthocomplete}
\cite[Definition 6.1]{COEA}, but this need not be the case for a Dedekind
orthocomplete GEA; hence we shall also stipulate that a DGEA is centrally
orthocomplete.

\section{Generalized effect algebras} \label{sc:GEAs} 

\noindent In this article, we use the abbreviation `iff' for `if and
only if,' the symbol $:=$ means `equals by definition,' and $\Nat:=
\{1,2,3,...\}$ is the set of natural numbers.

In what follows \emph{we assume that $(E,\oplus,0)$, or simply $E$ for short,
is a generalized effect algebra {\rm(}GEA{\rm)} with orthosummation $\oplus$ and
zero element $0$} \cite[Definition 2.1]{Zdenka99}. This means that $\oplus$ is a
partially defined binary operation on $E$, $0\in E$, and the following conditions
hold for all $d,e,f\in E$: (GEA1 \emph{Commutativity}) If $e\oplus f$ is defined,
then $f\oplus e$ is defined and $e\oplus f=f\oplus e$. (GEA2 \emph{Associativity})
If $d\oplus(e\oplus f)$ is defined, then $(d\oplus e)\oplus f$ is defined and
$d\oplus(e\oplus f)=(d\oplus e)\oplus f$. (GEA3 \emph{Neutral Element}) $e
\oplus 0=e$. (GEA4 \emph{Cancellation}) If $d\oplus e$ and $d\oplus f$ are
defined, then $d\oplus e=d\oplus f\Rightarrow e=f$. (GEA5 \emph{Positivity})
$e\oplus f=0\Rightarrow e=f=0$.

In general, elements of the GEA $E$ will be denoted by $d,e,f,p,q,s,$ and
$t$, with or without subscripts. By definition, $e$ and $f$ are
\emph{orthogonal}, in symbols $e\perp f$, iff $e\oplus f$ is defined.
If we write $e\oplus f$ without stipulating that $e\perp f$, the
condition $e\perp f$ is implicitly understood. If there exists
$d$ such that $e\oplus d=f$, we say that $e$ is a \emph{subelement}
of $f$, or that $f$ \emph{dominates} $e$, in symbols $e\leq f$.
It turns out that the GEA $E$ is partially ordered by the relation
$\leq$, which satisfies the following version of cancellation: $d
\oplus e\leq d\oplus f\Rightarrow e\leq f$. If $e\leq f$, then the
element $d$ such that $e\oplus d=f$ is uniquely determined by cancellation,
and we define $f\ominus e:=d$.

Evidently, $0$ is the smallest element in $E$. If there is a largest element
in $E$, it is often denoted by $1$ or $u$, it is called the \emph{unit},
and $E$ is called an \emph{effect algebra} (EA) \cite{FandB}. Informally,
we may think of a GEA as an EA ``possibly without a unit." If $E$ is an EA,
then the \emph{orthosupplement} of $e$, denoted and defined by $e
\sp{\perp} :=1\ominus e$, exists for all $e\in E$; moreover, $e\oplus e
\sp{\perp}=1$, and $e\perp f\Leftrightarrow e\leq f\sp{\perp}$. Every EA
$E$ satisfies both of the conditions in the following definition. (The
second condition follows from the observation that $e\perp f\sp{\perp}
\Leftrightarrow e\leq (f\sp{\perp})\sp\perp=f$.)

\begin{definition} \label{df:dirandOO}
(1) $E$ is (\emph{upward}) \emph{directed} iff, for every $e,f\in E$,
there exists $d\in E$ with $e,f\leq d$. (2) $E$ is \emph{orthogonally
ordered} iff, for every $e,f\in E$, the condition $d\perp f
\Rightarrow d\perp e$ for every $d\in E$ implies that $e\leq f$.
\end{definition}

Any set $E$ with three or more elements can be organized into a GEA that
is not an EA in a completely trivial way simply by choosing any element
in $E$, calling it $0$, and stipulating that, for $e,f\in E$, $e\perp f$
iff $e=0$ or $f=0$ (or both). Such a triviality can be ruled out by
requiring, for instance, either that $E$ is directed, or that $E$ is
orthogonally ordered.

A \emph{GEA-morphism} from $E$ to a second GEA $F$ is a mapping from
$E$ to $F$ that preserves orthogonality and orthosums. Such a GEA-morphism
automatically preserves $0$, partial order, and differences $p\ominus q$
for $q\leq p$. If $E$ and $F$ happen to be EAs, then an \emph{EA-morphism}
is defined to be a GEA-morphism that preserves the unit elements. A
\emph{GEA-endomorphism} is a GEA-morphism from a GEA into itself, and a
\emph{GEA-isomorphism} is a bijective GEA-morphism with an inverse that is
also a GEA-morphism.

We write an existing supremum (least upper bound) and an existing infimum
(greatest lower bound) of a family $(e\sb{i})\sb{i\in I}$ in $E$ as
$\bigvee\sb{i\in I}e\sb{i}$ and $\bigwedge\sb{i\in I}e\sb{i}$, respectively.
If $e,f\in E$, then the supremum (or join) of $e$ and $f$ in $E$, if it
exists, is written as $e\vee f$, and the infimum (or meet) of $e$ and $f$
in $E$, if it exists, is written as $e\wedge f$. If we write an equation
involving $e\vee f$ or $e\wedge f$ we understand that the existence
thereof is assumed. Two elements $e$ and $f$ are \emph{disjoint} iff
$e\wedge f=0$, i.e., iff the only common subelement of $e$ and $f$ is $0$.
The GEA $E$ is \emph{lattice ordered}, or simply a \emph{lattice}, iff
$e\vee f$ and $e\wedge f$ exist for all $e,f\in E$. Obviously, a lattice
ordered GEA is directed. A subset $S$ of $E$ is \emph{sup/inf-closed} in
$E$ iff all existing suprema and infima in $E$ of nonempty families in
$S$ belong to $S$.

Owing to GEA2, we may omit parentheses in the expression $d\oplus(e\oplus f)=
(d\oplus e)\oplus f$, and write simply $d\oplus e\oplus f$. For elements $e_1,e_2,
\ldots, e_n$, we define $e_1\oplus \cdots \oplus e_n$ by recurrence: $e_1\oplus
e_2\oplus\cdots\oplus e_n:=(e_1\oplus e_2\oplus\cdots\oplus e_{n-1})\oplus e_n$
iff all the required sub-orthosums exist. The orthogonality and the resulting
orthosum of a finite sequence in $E$ is independent of the order of its elements,
so we can define the orthogonality and the orthosum of a finite orthogonal
family in the obvious way. We understand that the empty family $(e\sb{i})\sb
{i\in\emptyset}$ is orthogonal with orthosum $0$.

Consider an arbitrary family $(e\sb{i})\sb{i\in I}$ in $E$. If $F$ is a finite
subset of the indexing set $I$, and if the finite subfamily $(e\sb{i})\sb
{i\in F}$ is orthogonal, then $\oplus\sb{i\in F}e\sb{i}$ is called a \emph
{finite partial orthosum} of $(e\sb{i})\sb{i\in I}$. By definition, $(e\sb{i})
\sb{i\in I}$ is \emph{orthogonal} iff every finite subfamily is orthogonal, and
it is \emph{orthosummable} iff it is orthogonal and the supremum $s$ of all of
its finite partial orthosums exists, in which case its \emph{orthosum} is
$\oplus\sb{i\in I}e\sb{i}:=s$.\footnote{Note that our terminology differs from
the terminology in \cite[p. 445]{JePu02}.}

Let $e\in E$ and put $e\sb{i}:=e$ for $i\in\Nat$. If the infinite sequence
$(e\sb{i})\sb{i\in\Nat}$ is orthogonal, we say that $e$ has \emph{infinite
isotropic index} in $E$.  The GEA $E$ is said to be \emph{archimedean} iff
$0$ is the only element in $E$ with infinite isotropic index.

The GEA $E$ is \emph{orthocomplete} iff every orthogonal family in $E$ is
orthosummable; it is \emph{Dedekind orthocomplete} iff, whenever all of
the finite partial orthosums of an orthogonal family in $E$ are bounded
from above, then the family is orthosummable. If $D\subseteq P\subseteq E$,
then $D$ is \emph{orthodense} in $P$ iff every element in $P$ is the
orthosum of an orthosummable family in $D$.

If $p\in E$, then the \emph{$p$-interval} $E[0,p]:=\{e\in E:0\leq e\leq p\}$
is organized into an EA with unit $p$, with orthogonality relation $\perp\sb{p}$,
and with orthosummation $\oplus\sb{p}$ defined by $e\perp\sb{p}f$ iff $e\oplus
f\leq p$ and $e\oplus\sb{p}f:=e\oplus f$ iff $e\perp\sb{p}f$. The resulting
partial order on $E[0,p]$ is the restriction to $E[0,p]$ of the partial
order $\leq$ on $E$, but, unless $p$ is a so-called \emph{principal} element
of $E$, i.e., $e,f\leq p$ with $e\perp f$ implies that $e\oplus f\leq p$,
\cite[Definition 3.2]{GFP}, the $p$-orthosummation $\oplus\sb{p}$ is not
the restriction to $E[0,p]$ of $\oplus$ on $E$.

If $p$ is principal, then $p$ is \emph{sharp}, i.e., if $d\leq p$ and
$d\perp p$, then $d=0$. Indeed, given such a $d$, we have $d,p\leq p$
with $d\perp p$, whence $d\oplus p\leq p=p\oplus 0$, so $d=0$ by
cancellation.

Clearly, the $p$-interval $E[0,p]$ is sup/inf-closed in $E$, and if a
nonempty family $(e\sb{i})\sb{i\in I}\subset E[0,p]$ has an infimum $e$
in the EA $E[0,p]$, then $e$ is the infimum in $E$ of $(e\sb{i})\sb{i
\in I}$; however, if the supremum $s$ of $(e\sb{i})\sb{i\in I}$ exists
in the EA $E[0,p]$, then $s$ need not be the supremum of $(e\sb{i})
\sb{i\in I}$ in $E$.

If $p\in E$, then a family $(e\sb{i})\sb{i\in I}$ in $E[0,p]$ is said to be
\emph{$p$-orthogonal} (respectively, \emph{$p$-orthosummable}) iff it is
orthogonal (respectively, orthosummable) in the EA $E[0,p]$. If $(e\sb{i})
\sb{i\in I}$ is $p$-orthosummable, then its orthosum in the EA $E[0,p]$ is
called its \emph{$p$-orthosum}. See \cite[Theorem 6.4]{CenGEA} for properties
of $p$-orthogonal families in $E$.

Let $S$ be a nonempty subset of $E$. We say that $S$ is a \emph{sub-GEA}
of $E$ iff, for all $s,t\in S$, (1) $s\perp t\Rightarrow s\oplus t\in S$ and
(2) $s\leq t\Rightarrow t\ominus s\in S$. If $S$ is a sub-GEA of $E$, then
$0\in S$ and $S$ is a GEA in its own right under the restriction to $S$ of
$\oplus$ on $E$. The subset $S$ is an \emph{order ideal} in $E$ iff,
whenever $s\in S$, $t\in E$, and $t\leq s$, it follows that $t\in S$.
By an \emph{ideal} in $E$, we mean an order ideal $S$ that satisfies
condition (1) above. Evidently, every ideal in $E$ is a sub-GEA of $E$,
and a $p$ interval $E[0,p]$ is an ideal iff $p$ is principal.

We say that $E$ is the \emph{direct sum} of the ideals $H\sb{1},H\sb{2},...,
H\sb{n}$ in $E$, in symbols $E=H\sb{1}\oplus H\sb{2}\oplus\cdots\oplus
H\sb{n}$ iff (1) if $h\sb{i}\in H\sb{i}$ for $i=1,2,...,n$, then
$h\sb{1},h\sb{2},...,h\sb{n}$ is an orthogonal sequence in $E$, and
(2) every element $e\in E$ decomposes uniquely into ``coordinates"
$e=e\sb{1}\oplus e\sb{2}\oplus\cdots\oplus e\sb{n}$ with $e\sb{i}\in
H\sb{i}$ for $i=1,2,...,n$. For such a direct sum, all GEA-calculations
on $E$ can be carried out ``coordinatewise" in the obvious sense. If $H$
is an ideal in $E$, we say that $H$ is a \emph{direct summand}\footnote{In
the literature \cite{Je00}, \cite[Lemma 2.7]{ExoCen} direct summands are
also called \emph{central ideals}.} of $E$ iff there is an ideal $K$ in $E$,
called a \emph{complementary direct summand} of $H$, such that $E=H\oplus K$.
It is not difficult to show that, if $H$ is a direct summand of $E$, then $H$
has a unique complementary direct summand $K$ in $E$.

\begin{remarks} \label{rm:altcentdef}
If $c\in E$, then by \cite[Definition 1.9.11]{DvPuTrends}, $c$ is
\emph{central} in $E$ iff (1) every element $e\in E$ can be written
uniquely as $e=e\sb{1}\oplus e\sb{2}$ with $e\sb{1}\leq c$ and $e\sb{2}
\perp c$, (2) $c$ is principal, and (3) if $p,q\in E$ with $p\perp q$,
then $p,q\perp c\Rightarrow p\oplus q\perp c$. By \cite[Lemma 4.7]
{CenGEA}, the uniqueness condition in (1) is automatic in the presence
of (2). The \emph{center} of $E$, denoted by $\Gamma(E)$, is defined to
be the set of all central elements of $E$.

According to \cite[Theorem 4.2]{ExoCen} (iii), $c\in\Gamma(E)$ iff
$E[0,c]$ is a direct summand of $E$. As a consequence of \cite
[Theorem 6.3]{PVRiesz}, if $c\in\Gamma(E)$, then the complementary direct
summand of $E[0,c]$ is the ideal $\{f\in E:f\perp c\}$. Obviously, every
central element in $E$ is principal. It turns out that $\Gamma(E)$ is a
lattice-ordered sub-GEA of $E$ and, as such, it is a generalized boolean
algebra \cite[Theorem 4.6 (iv)]{ExoCen}. \hspace{\fill} $\square$
\end{remarks}

\section{The Exocenter and Hull Systems}  

\noindent The \emph{exocenter} of the GEA $E$, in symbols $\GEX(E)$, is the set
of all mappings $\pi\colon E\to E$ such that, for all $e,f\in E$: (EXC1) $\pi
\colon E\to E$ is a \emph{GEA-endomorphism} of $E$, i.e., if $e\perp f$, then
$\pi e\perp\pi f$ and $\pi(e\oplus f)=\pi e\oplus\pi f$. (EXC2) $\pi$ is
\emph{idempotent}, i.e., $\pi(\pi e)=\pi e$. (EXC3) $\pi$ is \emph{decreasing},
i.e., $\pi e\leq e$. (EXC4) $\pi$ satisfies the following \emph{orthogonality
condition}: if $\pi e=e$ and $\pi f=0$, then $e\perp f$ \cite[\S 3]{ExoCen}.

If $\pi\in\GEX(E)$, define $\pi\,'\colon E\to E$ by $\pi\,' e :=e\ominus\pi e$
for all $e\in E$. Then $\pi\in\GEX(E)\Rightarrow\pi\,'\in\GEX(E)$. The composition
of mappings $\pi,\xi\in\GEX(E)$ is commutative, i.e., $\pi\circ\xi=\xi\circ\pi$, and
$\GEX(E)$ is partially ordered by $\pi\leq\xi$ iff $\pi\circ\xi=\pi$. In fact,
$\GEX(E)$ is a boolean algebra under $\leq$, with $\pi\mapsto\pi\,'$ as the boolean
complementation, the smallest element in $\GEX(E)$ is the zero mapping $0$ on $E$, the
largest element is the identity mapping $1$ on $E$, and for $\pi,\xi\in\GEX(E)$, the
infimum of $\pi$ and $\xi$ in $\GEX(E)$ is given by $\pi\wedge\xi=\pi\circ\xi=\xi
\circ\pi$. The infimum and supremum work pointwise, i.e., $(\pi\wedge\xi)e=\pi e
\wedge\xi e$ and $(\pi\vee\xi)e=\pi e\vee\xi e$ for all $e\in E$. We say that $\pi$
and $\xi$ are \emph{disjoint} iff $\pi\wedge\xi=0$.

There is a bijective correspondence $\pi\leftrightarrow H$ between mappings
$\pi\in\GEX(E)$ and direct summands $H$ of $E$ given by $H=\pi(E):=\{\pi e:e
\in E\}$. If $H=\pi(E)$, then the complementary direct summand of $H$ is
$K:=\pi\,'(E)$. If $\pi,\xi\in\GEX(E)$, then $\pi\leq\xi$ iff $\pi(E)
\subseteq\xi(E)$. If $\pi\sb{1},\pi\sb{2},...,\pi\sb{n}\in\GEX(E)$ and
$H\sb{i}:=\pi\sb{i}(E)$ for $i=1,2,...,n$, then $E=H\sb{1}\oplus H\sb{2}
\oplus\cdots\oplus H\sb{n}$ iff $\pi\sb{1},\pi\sb{2},...,\pi\sb{n}$
is a pairwise disjoint sequence in $\GEX(E)$ and $\pi\sb{1}\vee\pi\sb{2}
\vee\cdots\pi\sb{n}=1$. By definition, $E$ is \emph{irreducible} iff
$\GEX(E)=\{0,1\}$, i.e., iff $\{0\}$ and $E$ itself are the only direct
summands of $E$.

A family $(e\sb{i})\sb{i\in I}$ is said to be \emph{$\GEX$-orthogonal}
iff there exists a pairwise disjoint family $(\pi\sb{i})\sb{i\in I}$ in
$\GEX(E)$ such that $e\sb{i}=\pi\sb{i}e\sb{i}$ for all $i\in I$ \cite
[Definition 6.1]{ExoCen}. If $(e\sb{i})\sb{i\in I}$ is $\GEX$-orthogonal,
then it is orthogonal, and it is orthosummable iff $s:=\bigvee\sb{i\in I}
e\sb{i}$ exists in $E$, in which case $\oplus\sb{i\in I}e\sb{i}=s$ \cite
[Lemma 6.2]{ExoCen}. According to  \cite[Definition 6.4]{ExoCen}, $E$ is
a \emph{centrally orthocomplete GEA} (COGEA) iff (CO1) every $\GEX$-orthogonal
family in $E$ is orthosummable in $E$, and (CO2) if an element of $E$ is
orthogonal to each member of a $\GEX$-orthogonal family in $E$, then it
is orthogonal to the orthosum of the family. If $E$ is a COGEA,  then both
$\GEX(E)$ and $\Gamma(E)$ are complete boolean algebras \cite[Theorem 6.8]
{ExoCen}, \cite[Theorem 7.5 (iii)]{CenGEA}.

If $E$ is an EA that satisfies CO1, then condition CO2 is automatic. Also,
for  an EA $E$ with unit $1$, the boolean algebra $\GEX(E)$ is isomorphic
to the center $\Gamma(E)$ of $E$ under the mapping $\pi\leftrightarrow c$,
where $c:=\pi 1$ is the largest element in $\pi(E)$.

A \emph{hull system} on the GEA $E$ \cite[\S 7]{ExoCen}, \cite[\S 7]{HDTD} is
a family $(\eta\sb{e})\sb{e\in E}\subseteq\GEX(E)$ indexed by elements $e\in E$
such that, for all $e,f\in E$, (HS1) $\eta\sb{0}=0$, (HS2) $e\in E\Rightarrow e
=\eta\sb{e}e$, and (HS3) $e,f\in E\Rightarrow\eta\sb{\eta\sb{e}f}=\eta\sb{e}
\circ\eta\sb{f}=\eta\sb{e}\wedge\eta\sb{f}$.

A subset $\Theta$ of the exocenter $\GEX(E)$ is called \emph{hull determining}
(HD) iff (HD1) for each $e\in E$, there is a smallest mapping in the set
$\{\theta\in\Theta:\theta e=e\}$, and (HD2) $\theta,\xi\in\Theta\Rightarrow
\theta\wedge\xi\,'\in\Theta$ \cite[Definition 8.1]{HDTD}.  If $\Theta$ is an HD
subset of $\GEX(E)$, and if, for $e\in E$, $\eta\sb{e}$ is the smallest mapping
$\theta\in\Theta$ such that $\theta e=e$, then  $(\eta\sb{e})\sb{e\in E}$ is a hull
system on $E$, called the \emph{hull system determined by $\Theta$} \cite[Theorem
8.2]{HDTD}. If $(\eta\sb{e})\sb{e\in E}$ is a hull system on $E$, then the
\emph{$\eta$-exocenter} of $E$, $\Theta\sb{\eta}(E):=\{\eta\sb{e}:e\in E\}$, is a
hull determining subset of $\GEX(E)$ and the hull system that it determines is
$(\eta\sb{e})\sb{e\in E}$ itself \cite[Theorem 8.4]{HDTD}.

If $E$ is a COGEA, then there is a special hull system $(\gamma\sb{e})\sb{e\in E}$
on $E$, called the \emph{exocentral cover system}, such that, for each $e\in E$,
$\gamma\sb{e}$ is the smallest mapping $\pi\in\GEX(E)$ such that $\pi e=e$. \cite
[\S8]{ExoCen}, \cite[\S 5]{CenGEA}. In other words, $\gamma\sb{e}(E)$ is the
smallest direct summand of $E$ that contains the element $e$.

Let $(\eta\sb{e})\sb{e\in E}$ be a hull system on $E$. Then a family $(e\sb{i})
\sb{i\in I}$ in $E$ is called \emph{$\eta$-orthogonal} iff $(\eta\sb{e\sb{i}})
\sb{i\in I}$ is a pairwise disjoint family in $\GEX(E)$. Every $\eta$-orthogonal
family is $\GEX$-orthogonal. If $E$ is a COGEA, then a family is $\gamma$-orthogonal
iff it is $\GEX$-orthogonal; moreover every $\eta$-orthogonal family is
orthosummable and its orthosum is its supremum.

The following definition and theorem will be useful in our development of the
dimension theory for a GEA (Sections \ref{sc:DGEA}--\ref{sc:I/II/III} below).

\begin{definition}[{\cite[Definition 11.4]{HDTD}}] \label{df:etaTD}
If $E$ is a COGEA, $(\eta\sb{e})\sb{e\in E}$ is a hull system on $E$,
and $T\subseteq E$, we define $[T]\sb{\eta}$ to be the set of all
orthosums (suprema) of $\eta$-orthogonal families in $T$ and we define
$T\sp{\eta}:=\{\eta\sb{e}t:e\in E, t\in T\}$. The set $T$ is said
to be \emph{$\eta$-type determining} ($\eta$TD) iff $T=[T]\sb{\eta}=
T\sp{\eta}$; it is said to be \emph{strongly $\eta$-type determining}
($\eta$STD) iff $T$ is an order ideal and $T=[T]\sb{\eta}$.
\end{definition}

\noindent Clearly, if $T$ is $\eta$STD, then it is $\eta$TD.

\begin{theorem} [{\cite[Theorem 12.3 (i)]{HDTD}}] \label{th:tStar}
If $E$ is a COGEA, $(\eta\sb{e})\sb{e\in E}$ is a hull system on $E$,
and $T$ is an $\eta$TD subset of $E$, then there exists $t\sp{\ast}
\in T$ such that $\eta\sb{t\sp{\ast}}$ is the largest mapping in
the set $\{\eta\sb{t}:t\in T\}$.
\end{theorem}

As per part (i) of the next definition, a hull system $(\eta\sb{e})
\sb{e\in E}$ on $E$ determines an equivalence relation $\sim\sb{\eta}$
on $E$ and it turns out that $\sim\sb{\eta}$ has many of the features
of a dimension equivalence relation (DER) on $E$ (Section \ref{sc:DGEA}
below); in fact, under suitable conditions, it is a DER (Example
\ref{ex:etasimDER} below). Also, see Lemma \ref{lm:simeta=sim} (ii).

\begin{definition} \label{df:etasim}
Let $(\eta\sb{e})\sb{e\in E}$ be a hull system on $E$ and let $p\in E$. Then:
\begin{enumerate}
\item[(1)] For $e,f\in E$, $e\etasim f\Leftrightarrow\eta\sb{e}=\eta\sb{f}$.
\smallskip
\item[(2)] $p$ is an \emph{$\eta$-monad}\footnote{A different but equivalent
 definition of an $\eta$-monad is given in \cite[Definition 10.2 (2)]{HDTD}---see
 \cite[Theorem 10.9 (vii)]{HDTD}.} iff, for all $e\in E[0,p]$, $e\etasim p
 \Rightarrow e=p$.
\item[(3)] $p$ is an \emph{$\eta$-dyad} iff there are orthogonal elements
 $e,f\in E[0,p]$ such that $p=e\oplus f$ and $e\etasim f$.
\item[(4)] $(\eta\sb{e})\sb{e\in E}$ is \emph{divisible} iff, whenever $p,s,t
 \in E$ with $s\perp t$ and $p\etasim(s\oplus t)$, then there exist $e,f
 \in E[0,p]$ such that $e\perp f$, $p=e\oplus f$, $e\etasim s$, and $f\etasim t$.
\end{enumerate}
\end{definition}

\begin{lemma} \label{lm:divisible}
Let $(\eta\sb{e})\sb{e\in E}$ be a hull system on $E$ and let $p,s,t
\in E$ with $s\perp t$ and $p\sim\sb{\eta}(s\oplus t)$. Then the following
conditions are equivalent{\rm:} {\rm(i)} There exist $e,f\in E[0,p]$ with
$e\perp f$, $p=e\oplus f$, $e\sim\sb{\eta}s$, and $f\sim\sb{\eta}t$.
{\rm(ii)} $(\eta\sb{s}\wedge\eta\sb{t})p$ is an $\eta$-dyad in $E$.
\end{lemma}

\begin{proof}
(i) $\Rightarrow$ (ii). Assume (i), put $a:=(\eta\sb{s}\wedge\eta\sb{t})e$
and put $b:=(\eta\sb{s}\wedge\eta\sb{t})f$. Then
\[
(\eta\sb{s}\wedge\eta\sb{t})p=(\eta\sb{s}\wedge\eta\sb{t})(e\oplus f)=
 (\eta\sb{s}\wedge\eta\sb{t})e\oplus(\eta\sb{s}\wedge\eta\sb{t})f=
 a\oplus b.
\]
By \cite[Theorem 7.4 (xii)]{HDTD}, $\eta\sb{a}=\eta\sb{s}\wedge\eta\sb{t}
\wedge\eta\sb{e}=\eta\sb{s}\wedge\eta\sb{t}\wedge\eta\sb{s}=\eta\sb{s}
\wedge\eta\sb{t}$. Similarly, $\eta\sb{b}=\eta\sb{s}\wedge\eta\sb{t}$, so
$a\sim\sb{\eta}b$, whence $(\eta\sb{s}\wedge\eta\sb{t})p$ is an $\eta$-dyad
in $E$.

(ii) $\Rightarrow$ (i). Assume (ii). Then there exist $a,b\in E$ with
\setcounter{equation}{0}
\begin{equation} \label{eq:divisible01}
(\eta\sb{s}\wedge\eta\sb{t})p=a\oplus b\text{\ and\ }a\sim\sb{\eta}b.
\end{equation}
Thus, by \cite[Theorem 7.4 (ix) and (xii)]{HDTD},
\begin{equation} \label{eq:divisible02}
\eta\sb{a}=\eta\sb{b}=\eta\sb{a}\vee\eta\sb{b}=\eta\sb{a\oplus b}=\eta
 \sb{(\eta\sb{s}\wedge\eta\sb{t})p}=\eta\sb{s}\wedge\eta\sb{t}\wedge\eta
 \sb{p}.
\end{equation}
Put $q:=(\eta\sb{s}\wedge(\eta\sb{t})')p$ and $r:=((\eta\sb{s})'\wedge
\eta\sb{t})p$. Then we have
\[
\eta\sb{p}=\eta\sb{s\oplus t}=\eta\sb{s}\vee\eta\sb{t}=(\eta\sb{s}
 \wedge(\eta\sb{t})')\vee((\eta\sb{s})'\wedge\eta\sb{t})\vee(\eta\sb{s}
 \wedge\eta\sb{t}),
\]
 whence by (\ref{eq:divisible01}),
\begin{equation} \label{eq:divisible03}
p=\eta\sb{p}p=(\eta\sb{s}\wedge(\eta\sb{t})')p\oplus((\eta\sb{s})'
 \wedge\eta\sb{t})p\oplus(\eta\sb{s}\wedge\eta\sb{t})p=q\oplus r
 \oplus a\oplus b
\end{equation}
\[
=(q\oplus a)\oplus(r\oplus b)=e\oplus f\text{ where\ }e:=q\oplus a
 \text{\ and\ }f:=r\oplus b.
\]
Thus, by \cite[Theorem 7.4 (xiv)]{HDTD} and the fact that $\eta\sb{p}
=\eta\sb{s}\vee\eta\sb{t}\geq\eta\sb{s},\eta\sb{t}$, it follows that
\begin{equation} \label{eq:divisible04}
\eta\sb{e}=\eta\sb{q}\vee\eta\sb{a}=(\eta\sb{s}\wedge(\eta\sb{t})'
 \wedge\eta\sb{p})\vee(\eta\sb{s}\wedge\eta\sb{t}\wedge\eta\sb{p})
 =\eta\sb{s}\wedge\eta\sb{p}=\eta\sb{s}\text{\ and}
\end{equation}
\[
\eta\sb{f}=\eta\sb{r}\vee\eta\sb{b}=((\eta\sb{s})'\wedge\eta\sb{t}
 \wedge\eta\sb{p})\vee(\eta\sb{s}\wedge\eta\sb{t}\wedge\eta\sb{p})
 =\eta\sb{t}\wedge\eta\sb{p}=\eta\sb{t}.
\]
By (\ref{eq:divisible04}), we have $e\sim\sb{\eta}s$ and $f\sim\sb{\eta}t$;
hence (i) follows from (\ref{eq:divisible03}).
\end{proof}

\begin{theorem} \label{th:divisibility}
Let $(\eta\sb{e})\sb{e\in E}$ be a hull system on $E$. Then{\rm:}
{\rm(i)} $(\eta\sb{e})\sb{e\in E}$ is divisible iff, whenever
$p,s,t\in E$ with $s\perp t$ and $p\etasim(s\oplus t)$, it follows that
$(\eta\sb{s}\wedge\eta\sb{t})p$ is an $\eta$-dyad in $E$. {\rm(ii)}
If there are no nonzero $\eta$-monads in $E$, then $(\eta\sb{e})
\sb{e\in E}$ is divisible.
\end{theorem}

\begin{proof}
Part (i) follows directly from Lemma \ref{lm:divisible}. Also, by
\cite[Theorem 10.14]{HDTD}, given the hypotheses of (ii), every element
of $E$ is an $\eta$-dyad.
\end{proof}

\begin{theorem} \label{th:etaSKe}
Let $(\eta\sb{e})\sb{e\in E}$ be a hull system on $E$ and let $e,f,s,t
\in E$ with $e\oplus f=s\oplus t$. Then there exist $e\sb{1}, e\sb{2},
f\sb{1}, f\sb{2}\in E$ such that $e=e\sb{1}\oplus e\sb{2}$, $f=f\sb{1}
\oplus f\sb{2}$, $(e\sb{1}\oplus f\sb{1})\sim\sb{\eta}s$, and $(e\sb{2}
\oplus f\sb{2})\sim\sb{\eta}t$.
\end{theorem}

\begin{proof}
Put $p:=e\oplus f=s\oplus t$. Working in the EA $E[0,p]$ equipped with
the hull mapping $q\mapsto\eta\sb{q}p$ \cite[Theorem 9.3 (ii)]{HDTD},
and observing that $e\oplus\sb{p}f=e\oplus f=s\oplus t=s\oplus\sb{p}t$,
we can apply \cite[Theorem 3.17]{HandD} to obtain $e\sb{1},e\sb{2},
f\sb{1}, f\sb{2}\in E[0,p]$ such that
\[
e=e\sb{1}\oplus\sb{p}e\sb{2}=e\sb{1}\oplus e\sb{2},\ f=f\sb{1}\oplus
 \sb{p}f\sb{2}=f\sb{1}\oplus f\sb{2},\ \eta\sb{e\sb{1}\oplus f\sb{1}}p
 =\eta\sb{s}p,\text{\ and\ }\eta\sb{e\sb{2}\oplus f\sb{2}}p=\eta\sb{t}p.
\]
But, by \cite[Theorem 7.4 (xi)]{HDTD}, $\eta\sb{e\sb{1}\oplus f\sb{1}}p=
\eta\sb{s}p$ implies that $\eta\sb{e\sb{1}\oplus f\sb{1}}=\eta\sb{s}$,
and therefore $(e\sb{1}\oplus f\sb{1})\sim\sb{\eta}s$. Likewise, $(e\sb{2}
\oplus f\sb{2})\sim\sb{\eta}t$.
\end{proof}

\begin{lemma} \label{lm:etaSK4a}
Let $(\eta\sb{e})\sb{e\in E}$ be a hull system on $E$, let $e,f\in E$,
and put $e\sb{1}:=\eta\sb{f}e$, $f\sb{1}:=\eta\sb{e}f$. Then{\rm:}
{\rm(i)} $e\geq e\sb{1}\sim\sb{\eta}f\sb{1}\leq f$. {\rm(ii)}
$e\sb{1},f\sb{1}\not=0\Leftrightarrow\eta\sb{e}\wedge\eta\sb{f}\not=0$.
{\rm(iii)} $e\not\perp f\Rightarrow e\sb{1},f\sb{1}\not=0$.
\end{lemma}

\begin{proof}
As $\eta\sb{e\sb{1}}=\eta\sb{e}\wedge\eta\sb{f}=\eta\sb{f\sb{1}}$, we have
(i) and (ii), and (iii) follows from \cite[Theorem 7.4 (x)]{HDTD}.
\end{proof}

\section{An SK-Congruence} \label{sc:SK} 

\noindent A.N. Sherstnev \cite{Sh} and V.V. Kalinin \cite{K} launched the study of
orthomodular posets equipped with dimension equivalence relations as generalizations
of orthomodular dimension lattices \cite{Kalm86, Loom, Mae, Ram}.

In \cite[Definition 7.2]{HandD}, we generalized the notion of a \emph{Sherstnev-Kalinin
congruence} (SK-congruence) to an EA, and later we used an SK-congruence as a basis for our
definition of a \emph{dimension effect algebra} (DEA) \cite[Definition 7.14]{HandD}. In
this section, we further extend the definition of an SK-congruence to the generalized
effect algebra $E$.

\begin{definition} \label{df:SK-congruence}
A \emph{Sherstnev-Kalinin} (SK-) \emph{congruence}\footnote{An SK-congruence
is not necessarily a congruence in the sense of \cite[p. 448]{JePu02} since
conditions C3 and C4 may fail.} on the GEA $E$ is an equivalence relation
$\sim$ on $E$ such that, for all $e,f,p,s,t\in E$:
\begin{enumerate}
\item[(SK1)] $e\sim 0$ implies $e=0$.
\item[(SK2)] If $(e\sb{i})\sb{i\in I}$ and $(f\sb{i})\sb{i\in I}$
 are orthosummable families in $E$ and $e\sb{i}\sim f\sb{i}$ for all
 $i\in I$, then $\oplus\sb{i\in I}e\sb{i}\sim\oplus\sb{i\in I}
 f\sb{i}$.
\item[(SK3d)] \emph{If $p\sim s\oplus t$, then there exist $e, f
 \in E$ such that}
\[
 p=e\oplus f,\  e\sim s, \text{\ and\ } f\sim t.
\]
\item[(SK3e)]\emph{If $e\oplus f=s\oplus t$, then there are
 $e\sb{1},e\sb{2},f\sb{1},f\sb{2}\in E$ such that}
\[
 e=e\sb{1}\oplus e\sb{2},\,f=f\sb{1}\oplus f\sb{2},\,s\sim e\sb{1}
 \oplus f\sb{1}, \text{\ and\ } t\sim e\sb{2}\oplus f\sb{2}.
\]
\item[(SK4a)] \emph{If $e\not\perp f$, then there are nonzero $e\sb{1},
 f\sb{1}\in E$ such that $e\geq e\sb{1}\sim f\sb{1}\leq f$.}
\item[(SK4b)] \emph{If $e\not\leq f$, then there are nonzero $e\sb{1},
 d\sb{1}\in E$ such that $e\geq e\sb{1}\sim d\sb{1}\perp f$.}
\end{enumerate}
\end{definition}

The special case of SK2 for which $I$ is finite is called
\emph{finite additivity}, and SK3d is called (\emph{finite})
\emph{divisibility}. Note that conditions SK3d and SK3e, taken
together, are equivalent to the condition
\begin{enumerate}
\item[(SK3)]\emph{If $e\oplus f\sim s\oplus t$, then there are
 $e\sb{1},e\sb{2},f\sb{1},f\sb{2}\in E$ such that}
\[
 e=e\sb{1}\oplus e\sb{2},\,f=f\sb{1}\oplus f\sb{2},\,s\sim e\sb{1}
 \oplus f\sb{1}, \text{\ and\ } t\sim e\sb{2}\oplus f\sb{2}.
\]
\end{enumerate}

If $E$ is an EA, then SK4a $\Rightarrow$ SK4b, as can be seen by
replacing $f$ by $f\sp{\perp}$ in SK4a. More generally, if $E$ is
orthogonally ordered (Definition \ref{df:dirandOO} (ii)), then
SK4a $\Rightarrow$ SK4b. In $\S$6, we shall replace SK4a by a
stronger condition SK4a$'$.

\begin{definition} \label{df:SKDefs}
If $\sim$ is an SK-congruence on $E$, then, for all $d,e,f,h\in E${\rm:}
\begin{enumerate}
\item[(1)] If $e\sim f$, we say that $e$ and $f$ are \emph{equivalent}.
\item[(2)] $f\subsim e$ iff there exists $e\sb{1}\in E[0,e]$ such that $f
 \sim e\sb{1}$. If $f\subsim e$, we say that $f$ is \emph{sub-equivalent}
 to $e$.
\item[(3)] A subset $H$ of $E$ is \emph{hereditary} iff, whenever
 $h\in H$ and $e\subsim h$, it follows that $e\in H$.
\item[(4)] The elements $e$ and $f$ are \emph{related} iff there are
 nonzero elements $e\sb{1},f\sb{1}\in E$ such that $e\geq e\sb{1}
 \sim f\sb{1}\leq f$. If $e$ and $f$ are not related, they are
 \emph{unrelated}.
\item[(5)] $d$ is a \emph{descendent} of $e$ iff every nonzero subelement
 of $d$ is related to $e$.
\end{enumerate}
\end{definition}

\noindent Clearly, $H$ is a hereditary subset of $E$ iff $H$ is an order
ideal in $E$ and $e\sim h\in H\Rightarrow e\in H$. A straightforward argument
using SK3d (divisibility) shows that if $e$ and $f$ are related and $f
\sim d$, then $e$ and $d$ are related. It is also easy to see that if
$0\not=e\subsim f$, then $e$ and $f$ are related.

\begin{theorem} \label{th:etasimSK}
Suppose that $E$ is orthogonally ordered and $(\eta\sb{e})\sb{e\in E}$
is a hull system on $E$. Then $\etasim$ is an SK-congruence on $E$ iff
$(\eta\sb{e})\sb{e\in E}$ is divisible. Moreover, if $\etasim$ is an
SK-congruence on $E$, then elements $e,f\in E$ are related with respect
to $\etasim$ iff $\eta\sb{e}\wedge\eta\sb{f}\not=0$.
\end{theorem}

\begin{proof}(i) That $\etasim$ satisfies condition SK1 is obvious.
Suppose that $(e\sb{i})\sb{i\in I}$ and $(f\sb{i})\sb{i\in I}$ are
orthosummable families in $E$ with $s:=\oplus\sb{i\in I}e\sb{i}$,
$t:=\oplus\sb{i\in I}f\sb{i}$, and $e\sb{i}\etasim f\sb{i}$ for all
$i\in I$. Then by \cite[Theorem 7.4 (ix)]{HDTD}, $\eta\sb{s}=\bigvee
\sb{i\in I}\eta\sb{e\sb{i}}=\bigvee\sb{i\in I}\eta\sb{f\sb{i}}=
\eta\sb{t}$, so $s\etasim t$ and we have SK2. By Definition
\ref{df:etasim} (4), condition SK3d (divisibility) holds iff
$(\eta\sb{e})\sb{e\in E}$ is a divisible hull system. Condition SK3e
follows from Theorem \ref{th:etaSKe}, condition SK4a follows from
Lemma \ref{lm:etaSK4a}, and since $E$ is orthogonally ordered,
condition SK4b follows from SK4a. The last statement in the theorem
also follows from Lemma \ref{lm:etaSK4a}.
\end{proof}

\begin{assumption} \label{as:SK}
Henceforth we assume that $\sim$ is an SK-congruence on the GEA $E$.
\end{assumption}

\begin{definition} \label{df:splits}
A mapping $\pi\in\GEX(E)$ \emph{splits} $\sim$ iff whenever $e\in\pi(E)$,
$f\in\pi\,'(E)$, and $e\sim f$, then $e=f=0$. The subset of $\GEX(E)$
consisting of all $\pi\in\GEX(E)$ such that $\pi$ splits $\sim$ is denoted
by $\Sigma\sb{\sim}(E)$.
\end{definition}

\begin{lemma} \label{lm:piinSigmasim}
If $\pi\in\GEX(E)$, then the following conditions are mutually
equivalent{\rm:}
\begin{enumerate}
\item $\pi\in\Sigma\sb{\sim}(E)$ {\rm(}i.e., $\pi$ splits $\sim${\rm)}.
\item $f\sim e\in\pi(E)\Rightarrow f\in\pi(E)$.
\item $f\subsim e\in\pi(E)\Rightarrow f\in\pi(E)$, {\rm(}i.e.,
 $\pi(E)$ is hereditary{\rm)}.
\item If $e\in\pi(E)$ and $f\in\pi\,'(E)$, then $e$ is unrelated to $f$.
\end{enumerate}
\end{lemma}

\begin{proof}
(i) $\Rightarrow$ (ii). Assume (i) and the hypotheses of (ii).
Then $f=f\sb{1}\oplus f\sb{2}$ with $f\sb{1}\in\pi(E)$ and
$f\sb{2}\in\pi\,'(E)$. By divisibility, $e=e\sb{1}\oplus e\sb{2}$
with $e\sb{1}\sim f\sb{1}$ and $e\sb{2}\sim f\sb{2}$. As $e\sb{2}
\leq e\in\pi(E)$, it follows that $e\sb{2}\in\pi(E)$, whence
$e\sb{2}\sim f\sb{2}\in\pi\,'(E)$ implies that $e\sb{2}=f\sb{2}=0$,
and therefore $f=f\sb{1}\in\pi(E)$.

(ii) $\Rightarrow$ (iii). Assume (ii) and the hypothesis of
(iii). Then $f\sim e\sb{1}\leq e\in\pi(E)$ for some $e\sb{1}$,
whence $e\sb{1}\in\pi(E)$, and it follows from (ii) that
$f\in\pi(E)$.

(iii) $\Rightarrow$ (iv). Assume (iii) and suppose $e\in\pi(E)$,
$f\in\pi\,'(E)$ and $e\geq e\sb{1}\sim f\sb{1}\leq f$. Then,
$f\sb{1}\subsim e$, so $f\sb{1}\in\pi(E)$. But $f\sb{1}\leq f
\in\pi\,'(E)$, so $f\sb{1}\in\pi\,'(E)$, and it follows that $f\sb{1}
=0$, whence $e\sb{1}=0$, proving that $e$ and $f$ are unrelated.

(iv) $\Rightarrow$ (i). Assume (iv), let $e\in\pi(E)$ and let
$f\in\pi\,'(E)$ with $e\sim f$. If $e,f\not=0$, then $e$ and
$f$ are related. Therefore, $e=f=0$, and we have $\pi\in\Sigma\sb{\sim}(E)$.
\end{proof}

As a consequence of Lemma \ref{lm:piinSigmasim} (iii), if $\pi\in\Sigma
\sb{\sim}(E)$, then the direct summand $\pi(E)$ of $E$ is a hereditary
ideal in $E$ (i.e., it is both hereditary and an ideal); moreover, the
hereditary direct summands of $E$ are precisely those of the form $\pi(E)$
for $\pi\in\Sigma\sb{\sim}(E)$. The next theorem shows that, $\pi\in
\Sigma\sb{\sim}(E)$ iff, for the direct sum decomposition $E=\pi(E)
\oplus\pi\,'(E)$, the equivalence relation $\sim$ works ``coordinatewise."

\begin{theorem} \label{th:simCoordinatewise}
Let $\pi\in\GEX(E)$. Then $\pi\in\Sigma\sb{\sim}(E)$ iff, for all
$e,f\in E$,
\[
e\sim f\Leftrightarrow \pi e\sim\pi f\text{\ and\ }\pi\,'e\sim\pi\,'f.
\]
\end{theorem}

\begin{proof}
Assume that $\pi\in\Sigma\sb{\sim}(E)$ and suppose that $e\sim f$, i.e.,
$\pi e\oplus\pi\,' e\sim\pi f\oplus\pi\,' f$. By SK3, there exist
$a,b,c,d\in E$ such that $a\oplus b=\pi e$, $c\oplus d=\pi\,'e$, $a
\oplus c\sim \pi f$, and $b\oplus d\sim\pi\,'f$. Then $a\oplus c\sim
\pi f\in\pi(E)$, whence $a\oplus c\in\pi(E)$ by Lemma
\ref{lm:piinSigmasim}, and since $c\leq a\oplus c$, it follows that
$c\in\pi(E)$. But, $c\leq c\oplus d=\pi\,'e\in\pi\,'(E)$, so $c\in\pi
\,'(E)$, and it follows that $c=0$. Since $\pi\,'\in\Sigma\sb{\sim}(E)$,
a similar argument shows that $b=0$. Therefore $\pi e=a\oplus 0\sim\pi f$
and $\pi\,'e=0\oplus d\sim\pi\,'f$.

Conversely, assume that, for all $e,f\in E$,
\[
e\sim f\Leftrightarrow \pi e\sim\pi f\text{\ and\ }\pi\,'e\sim\pi\,'f.
\]
and suppose that $e\sim f\in\pi(E)$. Then $\pi\,'e\sim\pi\,'f=0$, whence
$\pi\,'e=0$, i.e., $e\in\pi(E)$. By Lemma \ref{lm:piinSigmasim}, we have
$\pi\in\Sigma\sb{\sim}(E)$.
\end{proof}

\begin{theorem} \label{th:Sigmasimboolean}
$\Sigma\sb{\sim}(E)$ is a boolean subalgebra of $\GEX(E)$.
\end{theorem}

\begin{proof}
Obviously, $0,1\in\Sigma\sb{\sim}(E)$ and if $\pi\in\Sigma\sb{\sim}(E)$, then
$\pi\,'\in\Sigma\sb{\sim}(E)$. Suppose that $\pi,\xi\in\Sigma\sb{\sim}(E)$ and
that $e,f\in E$ with $e\sim f\in(\pi\wedge\xi)(E)$. Then $e\sim f\in
\pi(E)$ and $e\sim f\in\xi(E)$, whence $e\in\pi(E)\cap\xi(E)=(\pi
\wedge\xi)(E)$ by Lemma \ref{lm:piinSigmasim}, and it follows from
Lemma \ref{lm:piinSigmasim} that $\pi\wedge\xi\in\Sigma\sb{\sim}(E)$.
Therefore $\pi\vee\xi=(\pi\,'\wedge\xi\,')'\in\Sigma\sb{\sim}(E)$, and
$\Sigma\sb{\sim}(E)$ is a boolean subalgebra of $\GEX(E)$.
\end{proof}

\begin{theorem}
Suppose that $E$ is orthogonally ordered and $(\eta\sb{e})\sb{e\in E}$
is a divisible hull system on $E$. Then $\Sigma\sb{\etasim}(E)$ is both
a boolean subalgebra of $\GEX(E)$ and a hull determining {\rm(HD)} set;
moreover the hull system determined by $\Sigma\sb{\etasim}(E)$ is
$(\eta\sb{e})\sb{e\in E}$ itself.
\end{theorem}

\begin{proof}
By Theorem \ref{th:etasimSK}, $\etasim$ is an SK-congruence on $E$, whence
$\Sigma\sb{\etasim}(E)$ is a boolean subalgebra of $\GEX(E)$ by Theorem
\ref{th:Sigmasimboolean}, and therefore it satisfies condition HD2. To
prove HD1, suppose $e\in E$, $s\in\eta\sb{e}(E)$, $t\in(\eta\sb{e})'(E)$,
and $s\etasim t$. Then $s=\eta\sb{e}s$, $\eta\sb{e}t=0$, and $\eta\sb{t}
=\eta\sb{s}=\eta\sb{\eta\sb{e}s}=\eta\sb{e}\wedge\eta\sb{s}=\eta\sb{e}
\wedge\eta\sb{t}=\eta\sb{\eta\sb{e}t}=0$, so $s=t=0$. This proves that
$\eta\sb{e}\in\Sigma\sb{\etasim}(E)$. To finish the proof, it will be
sufficient to show that $\eta\sb{e}$ is the smallest mapping $\pi\in
\Sigma\sb{\etasim}(E)$ such that $\pi e=e$. So assume that $\pi\in
\Sigma\sb{\etasim}(E)$ and $\pi e=e$. Choose any $f\in(\eta\sb{e}\wedge
\pi\,')(E)$ and put $s:=\eta\sb{f}e$, $t:=\eta\sb{e}f$.  Then $s\etasim
t$, $f=\eta\sb{e}f=t$, and $\pi\,'f=f$; hence $\pi s=\pi(\eta\sb{f}e)=\eta
\sb{f}(\pi e)=\eta\sb{f}e=s$ and $\pi\,'t=\pi\,'(\eta\sb{e}f)=\eta\sb{e}
(\pi\,'f)=\eta\sb{e}f=t$. Thus, $s\in\pi(E)$, $t\in\pi\,'(E)$, and
$s\etasim t$, whence $f=t=s=0$. But $f$ was an arbitrary element in
$(\eta\sb{e}\wedge\pi\,')(E)$, therefore $\eta\sb{e}\wedge\pi\,'
=0$ in the boolean algebra $\GEX(E)$, and it follows that $\eta\sb{e}
\leq\pi$.
\end{proof}

\begin{theorem} \label{th:SigmasimProps}
Suppose that $E$ is centrally orthocomplete {\rm(}i.e., a {\rm COGEA)}.
Then{\rm:}
\begin{enumerate}
\item $\Sigma\sb{\sim}(E)$ is a sup/inf-closed boolean subalgebra of the
 complete boolean algebra $\GEX(E)$.
\item $\Sigma\sb{\sim}(E)$ is a complete boolean algebra.
\item $\Sigma\sb{\sim}(E)$ is a hull determining {\rm(HD)} set and it
 determines the hull system $(\eta\sb{e})\sb{e\in E}$ given by
 $\eta\sb{e}:=\bigwedge\{\pi\in\Sigma\sb{\sim}(E):\pi e=e\}$ for all
 $e\in E$.
\item $e\in E\Rightarrow\eta\sb{e}\in\Sigma\sb{\sim}(E)$, i.e., $\Theta
 \sb{\eta}(E)\subseteq\Sigma\sb{\sim}(E)$.
\item If $e\in E$ and $\pi\in\Sigma\sb{\sim}(E)$, then $\pi e=0
 \Leftrightarrow \pi\wedge\eta\sb{e}=0$.
\item If $e,f\in E$, then there exists $\pi\in\Sigma\sb{\sim}(E)$ such
 that $e\in\pi(E)$ and $f\in\pi\,'(E)$ iff $\eta\sb{e}\wedge\eta\sb{f}=0$.
\item If $e,f\in E$ and $\eta\sb{e}\wedge\eta\sb{f}=0$, then $e$ is
 unrelated to $f$.
\end{enumerate}
\end{theorem}

\begin{proof}
By \cite[Theorem 6.8]{ExoCen}, $\GEX(E)$ is a complete boolean algebra
and by Theorem \ref{th:Sigmasimboolean}, $\Sigma\sb{\sim}(E)$ is a boolean
subalgebra of $\GEX(E)$. Let $(\pi\sb{i})\sb{i\in I}$ be a family of
mappings in $\Sigma\sb{\sim}(E)$ and put $\pi:=\bigwedge\sb{i\in I}\pi
\sb{i}$ {\rm(}the infimum in $\GEX(E)${\rm)}. Then, since infima in
$\GEX(E)$ can be calculated pointwise \cite[Theorem 6.9]{ExoCen}, we
have $\pi(E)=\bigcap\sb{i\in I}\pi\sb{i}(E)$. Suppose that $e,f\in E$
and $e\sim f\in\pi(E)$. Then, for every $i\in I$, $e\sim f\in\pi\sb{i}(E)$,
and since $\pi\sb{i}\in\Sigma\sb{\sim}(E)$, it follows from Lemma
\ref{lm:piinSigmasim} that $e\in\pi\sb{i}(E)$. Therefore, $e\in\bigcap
\sb{i\in I}\pi\sb{i}(E)=\pi(E)$, and Lemma \ref{lm:piinSigmasim} implies
that $\pi\in\Sigma\sb{\sim}(E)$, whence (by the de\,Morgan law) $\Sigma
\sb{\sim}(E)$ is sup/inf-closed in $\GEX(E)$. This proves (i), and (ii)
follows immediately from (i).

By \cite[Theorem 8.9]{HDTD} (which requires the hypothesis that $E$ is
a COGEA), $\Sigma\sb{\sim}(E)$ is HD, the hull system that it determines
is given by $\eta\sb{e}:=\bigwedge\{\pi\in\Sigma\sb{\sim}(E):\pi e=e\}$
for all $e\in E$, and we have (iii). As $\eta\sb{e}$ is the smallest
mapping in $\{\pi\in\Sigma\sb{\sim}(E):\pi e=e\}$, we also have (iv).

Assume the hypotheses of (v). Then $\pi\,'\in\Sigma\sb{\sim}(E)$, so
if $\pi e=0$, then $\pi\,'(e)=e$, and it follows from (iii) that
$\eta\sb{e}\leq\pi\,'$, i.e., $\pi\wedge\eta\sb{e}=0$. Conversely,
if $\pi\wedge\eta\sb{e}=0$, then $\pi e=\pi e\wedge e=\pi(e)\wedge
\eta\sb{e}e=(\pi\wedge\eta\sb{e})e=0$, proving (v).

To prove (vi), suppose that $\pi\in\Sigma\sb{\sim}(E)$, $e\in\pi(E)$,
and $f\in\pi\,'(E)$. Then $e=\pi e$ and $f=\pi\,'f$, whence $\eta\sb{e}
\leq\pi$ and $\eta\sb{f}\leq\pi\,'$, and it follows that $\eta\sb{e}
\wedge\eta\sb{f}\leq\pi\wedge\pi\,'=0$. Conversely, suppose that $\eta
\sb{e}\wedge\eta\sb{f}=0$, i.e., $\eta\sb{f}\leq(\eta\sb{e})'$. Put
$\pi:=\eta\sb{e}$. Then $e=\eta\sb{e}e=\pi e\in\pi(E)$, and as $\eta
\sb{f}\leq\pi\,'$, we have $f=\eta\sb{f}f\leq\pi\,'f\leq f$, so $f=
\pi\,'f\in\pi\,'(E)$. This proves (vi), and (vii) follows from (vi) and
Lemma \ref{lm:piinSigmasim} (iv).
\end{proof}

\begin{theorem} \label{th:pi&sim}
Suppose that $E$ is a COGEA, let $(\eta\sb{e})\sb{e\in E}$ be the hull
system determined by the HD set $\Sigma\sb{\sim}(E)$, let $e,f\in E$,
and let $\pi\in\Sigma\sb{\sim}(E)$. Then{\rm:} {\rm(i)} If $e\subsim f$,
then $\pi e\subsim\pi f$ and $\eta\sb{e}\leq\eta\sb{f}$. {\rm(ii)} $\eta
\sb{e}\leq\eta\sb{f}$ iff there exists $f\sb{1}\in E$ with $e\etasim
f\sb{1}\leq f$. {\rm(iii)} If $e\sim f$, then $\pi e\sim\pi f$ and $e\etasim f$.
\end{theorem}

\begin{proof}
Assume the hypotheses. Suppose that $e\subsim f$. Then there exists
$f\sb{1}\in E$ with $e\sim f\sb{1}\leq f$. By Theorem
\ref{th:simCoordinatewise}, $\pi e\sim\pi f\sb{1}$, and since $\pi\in
\GEX(E)$, $\pi f\sb{1}\leq\pi f$, whence $\pi e\subsim\pi f$.  If $\pi f=f$,
then $f\in\pi(E)$, and since $\pi(E)$ is hereditary, it follows that
$e\in\pi(E)$. Thus, $\pi f=f$ implies $\pi e=e$, and it follows that
\[
\eta\sb{e}=\bigwedge\{\pi\in\Sigma\sb{\sim}(E):\pi e=e\}\leq\bigwedge
 \{\pi\in\Sigma\sb{\sim}(E):\pi f=f\}=\eta\sb{f}.
\]
This proves (i), and (iii) follows from (i).

To prove (ii), suppose that $\eta\sb{e}\leq\eta\sb{f}$ and put $f\sb{1}
:=\eta\sb{e}f$. Then $f\sb{1}\leq f$ and $\eta\sb{f\sb{1}}=\eta\sb{e}
\wedge\eta\sb{f}=\eta\sb{e}$, so $e\etasim f\sb{1}$. Conversely, if
$e\etasim f\sb{1}\leq f$, then $\eta\sb{e}=\eta\sb{f\sb{1}}\leq\eta\sb{f}$.
\end{proof}

With the aid of \cite[Theorem 6.3]{CenGEA} and the hypothesis that $E$ is
Dedekind orthocomplete, the proof of the next lemma closely follows the
proof of \cite[Proposition 3]{Je02}.

\begin{lemma} \label{lm:simcancel}
If $E$ is Dedekind orthocomplete, then $e\sim e\oplus f\oplus d
\Rightarrow e\sim e\oplus f$.
\end{lemma}

\begin{theorem} \label{th:preorder}
If $E$ is Dedekind orthocomplete, then the relation $\subsim$ is a preorder
{\rm (}reflexive and transitive{\rm )}, on $E$ and the following
Cantor-Schr\"{o}der-Bernstein property holds{\rm:} $e\subsim f\subsim e
\Rightarrow e\sim f$.
\end{theorem}

\begin{proof}
Let $e,f,d\in E$. Obviously $e\subsim e$. Suppose $e\subsim f$ and
$f\subsim d$. Then there exist $f\sb{1}$, $d\sb{1}$ in $E$ with
$e\sim f\sb{1}\leq f\sim d\sb{1}\leq d$. By divisibility,
$d\sb{1}=d\sb{2}\oplus d\sb{3}$ with $d\sb{2}\sim f\sb{1}$ and
$d\sb{3}\sim(f\ominus f\sb{1})$; hence $e\sim d\sb{2}\leq d$, so
$e\subsim d$.

Suppose that $e\subsim f\subsim e$. The following proof that
$e\sim f$ is extracted from the proofs of \cite[Propositions 1 and 2]
{Je02}. There exist $e\sb{1},f\sb{1}\in E$ such that $e\sim f\sb{1}
\leq f\sim e\sb{1}\leq e$. As $e\sb{1}\sim f=f\sb{1}\oplus(f\ominus
f\sb{1})$, divisibility implies the existence of $e\sb{2}$ and
$e\sb{3}$ such that $e\sb{1}=e\sb{2}\oplus e\sb{3}$, $e\sb{2}\sim
f\sb{1}\sim e$ and $e\sb{3}\sim(f\ominus f\sb{1})$. Thus, $e\sb{2}
\sim e=e\sb{1}\oplus(e\ominus e\sb{1})=e\sb{2}\oplus e\sb{3}\oplus
(e\ominus e\sb{1})$; hence by Lemma \ref{lm:simcancel}, $e\sim
e\sb{2}\sim e\sb{2}\oplus e\sb{3}=e\sb{1}\sim f$.
\end{proof}

\begin{lemma} \label{le:AdditivityofSubEq}
If $E$ is Dedekind orthocomplete and if $(e\sb{i})\sb{i\in I}$ and
$(f\sb{i})\sb{i\in I}$ are orthosummable families in $E$ such that
$e\sb{i}\subsim f\sb{i}$ for all $i\in I$, then $\oplus\sb{i\in I}e
\sb{i}\subsim\oplus\sb{i\in I}f\sb{i}$.
\end{lemma}

\begin{proof}
As $e\sb{i}\subsim f\sb{i}$ for all $i\in I$, there is a family $(d
\sb{i})\sb{i\in I}$ with $e\sb{i}\sim d\sb{i}\leq f\sb{i}$ for all
$i\in I$. If $F$ is a finite subset of $I$, then $(f\sb{i})\sb{i\in F}$
is orthogonal with $\oplus\sb{i\in F}f\sb{i}\leq p:=\oplus\sb{i\in I}
f\sb{i}$, whence $(d\sb{i})\sb{i\in I}$ is orthogonal with $\oplus
\sb{i\in F}d\sb{i}\leq\oplus\sb{i\in F}f\sb{i}\leq p$, and it follows
from Dedekind orthocompleteness that $(d\sb{i})\sb{i\in I}$ is
orthosummable with $d:=\oplus\sb{i\in I}d\sb{i}\leq p$. Therefore by
SK2, $\oplus\sb{i\in I}e\sb{i}\sim d\leq p=\oplus\sb{i\in I}f\sb{i}$,
whence $\oplus\sb{i\in I}e\sb{i}\subsim\oplus\sb{i\in I}f\sb{i}$.
\end{proof}

\begin{theorem} \label{th:decompospq}
If $E$ is Dedekind orthocomplete and $p,q\in E$, then there are
orthogonal decompositions $p=p\sb{1}\oplus p\sb{2}$ and $q=q\sb{1}
\oplus q\sb{2}$ such that $p\sb{1}\sim q\sb{1}$ and $p\sb{2}$ is
unrelated to $q\sb{2}$.
\end{theorem}

\begin{proof}
As $E$ is Dedekind orthocomplete, the EAs $E[0,p]$ and $E[0,q]$ are
archimedean; hence we may invoke Zorn's lemma to produce a maximal
family of pairs $(e\sb{i},f\sb{i})\sb{i\in I}$ such that $(e\sb{i})
\sb{i\in I}\subseteq E[0,p]$ is $p$-orthogonal, $(f\sb{i})\sb
{i\in I}\subseteq E[0,q]$ is $q$-orthogonal, and $e\sb{i}\sim f
\sb{i}$ for all $i\in I$ (Cf. \cite[Remarks 4.3]{HDTD}. By
\cite[Theorem 6.4 (iv)]{CenGEA}, $(e\sb{i})\sb{i\in I}$ is both orthosummable
in $E$ and $p$-orthosummable in $E[0,p]$; moreover, its  orthosum $p\sb{1}
:=\oplus\sb{i\in I}e\sb{i}$ in $E$ coincides with its $p$-orthosum in
$E[0,p]$. Likewise for $(f\sb{i})\sb{i\in I}$ in $E[0,q]$ and for its
orthosum $q\sb{1}:=\oplus\sb{i\in I}e\sb{i}$, and we have $p\sb{1}\sim q
\sb{1}$ by SK2. Put $p\sb{2}:=p\ominus p\sb{1}$ and $q\sb{2}:=q\ominus q
\sb{1}$. Then $p\sb{2}$ is unrelated to $q\sb{2}$, else there would exist
nonzero $e\sb{0},f\sb{0}$ with $p\sb{2}\geq e\sb{0}\sim f\sb{0}
\leq q\sb{2}$, and $(e\sb{i},f\sb{i})\sb{i\in I}$ could be enlarged by
appending the pair $(e\sb{0},f\sb{0})$.
\end{proof}

\section{Hereditary Intervals and Invariant Elements}
\label{sc:HIetaInv} 

\noindent\emph{The assumption that $\sim$ is an SK-congruence on the
GEA $E$ remains in force.}

\begin{lemma} \label{lm:hereditary1}
Let $c\in E$ and suppose that the $c$-interval $E[0,c]$ is hereditary.
Then, for all $d\in E$, $d\wedge c=0\Rightarrow d\perp c$.
\end{lemma}

\begin{proof} Assume that $E[0,c]$ is hereditary. (i) Suppose $d
\wedge c=0$, but $d\not\perp c$. Then by SK4a, there exist nonzero
$d\sb{1},c\sb{1}\in E$ with $d\geq d\sb{1}\sim c\sb{1}\leq c$, whence
$d\sb{1}\subsim c$, so $0\not=d\sb{1}\leq c,d$, contradicting $d
\wedge c=0$.
\end{proof}

An element in a Loomis dimension lattice is defined to be
\emph{invariant} iff it is unrelated to its orthocomplement
\cite[p. 6]{Loom}, a condition which is easily adapted to the
GEA $E$ as in part (i) of the next lemma. (See Definition
\ref{df:invariant} below.)  Since every element in a Loomis
dimension lattice is sharp, the lemma generalizes \cite[Lemma 21]{Loom}.

\begin{lemma} [{Cf. \cite[Lemma 7.11]{JePu02}}] \label{lm:hereditary2}
If $c\in E$, then the following conditions are mutually equivalent{\rm:}
{\rm(i)} If $c\sb{1},f\in E$, then $c\geq c\sb{1}\sim f\perp c\Rightarrow
c\sb{1}=f=0$. {\rm(ii)} $c$ is sharp and $E[0,c]$ is hereditary.
{\rm(iii)} $c$ is sharp and for all $e\in E$, $e\sim c\Rightarrow
e\leq c$. {\rm(iv)} $c$ is sharp and $\{f\in E:f\perp c\}$ is hereditary.
\end{lemma}

\begin{proof}
(i) $\Rightarrow$ (ii) $\Rightarrow$ (iii). Assume (i). Then, if $f
\leq c$, and $f\perp c$, we have $c\geq f\sim f\perp c$, so $f=0$ by (i),
whence $c$ is sharp. Now suppose $d\sim c\sb{1}\leq c$ but $d\not
\in E[0,c]$. Then by SK4b, there exist nonzero $d\sb{1},f\in E$ such
that $d\geq d\sb{1}\sim f\perp c$. We have $c\sb{1}\sim d=d\sb{1}
\oplus(d\ominus d\sb{1})$, so by SK3d, there exist $c\sb{2},c\sb{3}$
such that $c\sb{1}=c\sb{2}\oplus c\sb{3}$ and $c\sb{2}\sim d\sb{1}$.
Thus, $c\sb{2}\sim d\sb{1}\sim f$, whence $c\geq c\sb{2}\sim f\perp c$,
and it follows from (i) that $f=0$, a contradiction. Therefore, $d
\in E[0,c]$. Obviously, (ii) $\Rightarrow$ (iii).

(iii) $\Rightarrow$ (i) Assume (iii), suppose that $c\sb{1},f\in E$
with $c\geq c\sb{1}\sim f\perp c$, and put $e:=f\oplus(c\ominus c
\sb{1})$. Then as $c=c\sb{1}\oplus(c\ominus c\sb{1})$, $f\sim c
\sb{1}$, and $(c\ominus c\sb{1})\sim (c\ominus c\sb{1})$, SK2 implies
that $e\sim c$, whence $f\leq f\oplus(c\ominus c\sb{1})=e\leq c$, and
since $c$ is sharp, $f=0$, so $c\sb{1}=0$ too. Thus, (i)
$\Leftrightarrow$ (ii) $\Leftrightarrow$ (iii).

(ii) $\Rightarrow$ (iv). Assume (ii) and suppose that $d,f\in E$ with $d
\subsim f\perp c$, but $d\not\perp c$. Then there exists $f\sb{1}$ such
that $d\sim f\sb{1}\leq f\perp c$, so $f\sb{1}\perp c$. By Lemma
\ref{lm:hereditary1}, there exists $d\sb{1}\in E$ with $0\not=d\sb{1}
\leq d,c$. Now $f\sb{1}\sim d=d\sb{1}\oplus(d\ominus d\sb{1})$, so by
SK3d, there exist $f\sb{2},f\sb{3}\in E$ such that $f\sb{1}=f\sb{2}
\oplus f\sb{3}$ and $f\sb{2}\sim d\sb{1}\leq c$. Thus $f\sb{2}\subsim c$,
so $f\sb{2}\leq c$. But $f\sb{2}\leq f\perp c$, whence $f\sb{2}\perp c$,
and as $c$ is sharp, $f\sb{2}=0$. Consequently, $d\sb{1}=0$, a
contradiction.

(iv) $\Rightarrow$ (i). Assume (iv) and suppose that $c\geq c\sb{1}
\sim f\perp c$. Then $c\sb{1}\perp c$, and since $c$ is sharp,
$c\sb{1}=0$, so $f=0$ as well.
\end{proof}

\begin{lemma} \label{lm:hereditary3}
Suppose that $c\in E$ satisfies any one, hence all of the conditions
in Lemma \ref{lm:hereditary2} and that every $e\in E$ decomposes in
$E$ as $e=e\sb{1}\oplus e\sb{2}$ with $e\sb{1}\leq c$ and $e\sb{2}
\perp c$. Then $c\in\Gamma(E)$.
\end{lemma}

\begin{proof}
Assume the hypotheses of the lemma. We claim that $c$ is principal.
Indeed, suppose that $p,q\in E$ with $p\perp q$ and $p,q\leq c$.
Then there exist $e\sb{1}, e\sb{2}\in E$ with $p\oplus q=e\sb{1}
\oplus e\sb{2}$, $e\sb{1}\leq c$ and $e\sb{2}\perp c$, and by SK3e,
there exist $a,b,v,w\in E$ with $b\leq a\oplus b=p\leq c$, $w\leq v
\oplus w=q\leq c$, $a\oplus v\sim e\sb{1}\leq c$, and $b\oplus w
\sim e\sb{2}\perp c$. As both $E[0,c]$ and $\{f\in E:f\perp c\}$
are hereditary, we have $a\oplus v\leq c$ and $b,w\leq b\oplus w
\perp c$. Therefore, since $b,w\leq c$ and $c$ is sharp, we have
$b=w=0$, and it follows that $p\oplus q=a\oplus v\leq c$.
This proves that $c$ is principal.

Now suppose that $f\sb{1},f\sb{2}\perp c$ with $f\sb{1}\perp f
\sb{2}$. We claim that $f\sb{1}\oplus f\sb{2}\perp c$. Indeed,
there exist $e\sb{1},e\sb{2}\in E$ with $f\sb{1}\oplus f\sb{2}=e
\sb{1}\oplus e\sb{2}$, $e\sb{1}\leq c$, and $e\sb{2}\perp c$. By
SK3e, there exist $s,t,x,y\in E$ with $s\leq s\oplus t=f\sb{1}
\perp c$, $x\leq x\oplus y=f\sb{2}\perp c$, $s\oplus x\sim e\sb{1}
\leq c$, and $t\oplus y\sim e\sb{2}\perp c$. Again, as both $E[0,c]$
and $\{f\in E:f\perp c\}$ are hereditary, we have $s,x\leq s\oplus x
\leq c$ and $t\oplus y\perp c$. Therefore, since $c$ is sharp,
$s=x=0$, and it follows that $f\sb{1}\oplus f\sb{2}=t\oplus y
\perp c$, as claimed.

By the hypotheses of the lemma, the fact that $c$ is principal, and the
result in the latter paragraph, we infer as per Remarks \ref{rm:altcentdef}
that $c\in\Gamma(E)$.
\end{proof}

Recall that if $c\in E$, then $c\in\Gamma(E)$ iff $E[0,c]$ is a direct
summand of $E$ (Remarks \ref{rm:altcentdef}), i.e., iff there exists a
uniquely determined mapping, \emph{which we shall denote by $\pi\sb{c}\in
\GEX(E)$}, such that $\pi\sb{c}(E)=E[0,c]$.

\begin{lemma} \label{lm:hereditary4}
Let $E$ be Dedekind orthocomplete. Then the following two conditions are
equivalent{\rm:}  {\rm(i)} $c$ is principal and $E[0,c]$ is hereditary.
{\rm(ii)} $c\in\Gamma(E)$ and $\pi\sb{c}\in\Sigma\sb{\sim}(E)$.
\end{lemma}

\begin{proof} (i) $\Rightarrow$ (ii). Assume (i).  We claim that
every $e\in E$ can be decomposed in $E$ as $e=e\sb{1}\oplus e\sb{2}$
with $e\sb{1}\leq c$ and $e\sb{2}\perp c$. Indeed, since $E$ is
Dedekind orthocomplete, \cite[Theorem 5.4 (iii)]{ExoCen} implies
that every chain (totally ordered set) in the bounded set $E[0,c]
\cap E[0,e]$ has a supremum; hence by Zorn's lemma, there is a
maximal element $e\sb{1}\in E[0,c]\cap E[0,e]$. Thus $e\sb{1}
\leq c,e$ and with $e\sb{2}:=e\ominus e\sb{1}$, we have $e=e
\sb{1}\oplus e\sb{2}$. We have to prove that $e\sb{2}\perp c$. Aiming
for a contradiction, we assume that $e\sb{2}\not\perp c$. Then by
Lemma \ref{lm:hereditary1}, there exists $0\not=p\leq e\sb{2},c$; hence,
with $e\sb{3}:=e\sb{2}\ominus p$, we have $e\sb{2}=p\oplus e\sb{3}$, and
therefore $e\sb{1}\oplus p\leq e\sb{1}\oplus p\oplus e\sb{3}=e\sb{1}
\oplus e\sb{2}=e$. Moreover, as $e\sb{1},p\leq c$ and $c$ is principal,
it follows that $e\sb{1}<e\sb{1}\oplus p\leq c$, contradicting the
maximality of $e\sb{1}$, and proving the claim.

Since $c$ is principal, it is sharp, so it satisfies condition (ii)
in Lemma \ref{lm:hereditary2}; hence, by the result above and
Lemma \ref{lm:hereditary3}, $c\in\Gamma(E)$. Therefore, since
$E[0,c]=\pi\sb{c}(E)$ is hereditary, $\pi\sb{c}\in\Sigma\sb{\sim}
(E)$ by the equivalence of (i) and (iii) in Lemma \ref{lm:piinSigmasim}.

(ii) $\Rightarrow$ (i).  Assume (ii). Then $c$ is principal since $c\in
\Gamma(E)$, and we have $\pi\sb{c}(E)=E[0,c]$. Also, by Lemma
\ref{lm:piinSigmasim}, $\pi\sb{c}\in\Sigma\sb{\sim}(E)$ implies that
$\pi\sb{c}(E)=E[0,c]$ is hereditary.
\end{proof}

\begin{lemma} [{Cf. \cite[Theorem 7.13]{JePu02}}] \label{lm:dir&invar}
Suppose that $c$ satisfies any one {\rm(}hence all{\rm)} of the conditions
in Lemma \ref{lm:hereditary2}. Then{\rm:} {\rm(i)} If $E$ is directed, then
$c\in\Gamma(E)$. {\rm(ii)} If $E$ is orthogonally ordered, then $c$ is
principal. {\rm (iii)} If $E$ is Dedekind orthocomplete and either directed or
orthogonally ordered, then $c\in\Gamma(E)$ and $\pi\sb{c}\in\Sigma\sb{\sim}(E)$.
\end{lemma}

\begin{proof}
Assume the hypothesis. (i) Suppose that $E$ is directed and let $e\in E$.
We claim that there exist $e\sb{1},e\sb{2}\in E$ with $e=e\sb{1}\oplus e
\sb{2}$, $e\sb{1}\leq c$, and $e\sb{2}\perp c$. Indeed, as $E$ is
directed, there exists $p\in E$ with $e,c\leq p$, and with $x:=p\ominus e$
and $y:=p\ominus c$, we have $p=e\oplus x=c\oplus y$. Thus, by SK3e, there
exist $e\sb{1},e\sb{2},s,t\in E$ such that $e=e\sb{1}\oplus e\sb{2}$,
$x=s\oplus t$, $e\sb{1}\oplus s\sim c$, and $e\sb{2}\oplus t\sim y\perp c$.
By parts (ii) and (iv) of Lemma \ref{lm:hereditary2}, $e\sb{1}\leq e\sb{1}
\oplus s\leq c$ and $e\sb{2}\leq e\sb{2}\oplus t\perp c$; hence $e\sb{2}\perp
c$, so (i) holds by Lemma \ref{lm:hereditary3}. (ii) Suppose that
$E$ is orthogonally ordered, let $p,q\in E[0,c]$ with $p\perp q$, but
suppose that $p\oplus q\notin E[0,c]$. Then there exists $h\in E$ with
$h\perp c$ and $h\not\perp p\oplus q$, whence by SK4a, there exist nonzero
$r,k\in E$ with $p\oplus q\geq r\sim k\leq h$. Since $r\leq p\oplus q$,
there exists $s\in E$ with $p\oplus q=r\oplus s$; hence by SK3e, there
exist $a,b,x,y\in E$ with $a\oplus b=p\leq c$, $x\oplus y=q\leq c$,
$a\oplus x\sim r$, and $b\oplus y\sim s$. As $a\oplus x\sim r\sim k
\leq h\perp c$, we infer from part (iv) of Lemma \ref{lm:hereditary2}
that $a\oplus x\perp c$. Thus $a\perp c$, $a\leq p\leq c$, $x\perp c$,
and $x\leq q\leq c$, whence, since $c$ is sharp, $a=x=0$. Therefore,
$r\sim a\oplus x=0$, we arrive at the contradiction $r=0$, and (ii)
is proved. Part (iii) follows from (i), (ii), and Lemma \ref{lm:hereditary4}.
\end{proof}

Since every element of a Loomis dimension lattice is principal, the
following definition extends the idea of an invariant element to
the GEA $E$.

\begin{definition} \label{df:invariant}
An element $c\in E$ is \emph{invariant} iff $c$ is principal and,
for all $e\sb{1},f\in E$, $e\geq e\sb{1}\sim f\perp c\Rightarrow
e\sb{1}=f=0$. We denote by $\Gamma\sb{\sim}(E)$ the set of all
invariant elements in $E$.
\end{definition}

A second notion of ``invariance" applies to any hull system
on $E$ \cite[Definition 7.1]{ExoCen}. In particular:

\begin{definition} \label{df:etainvar}
Let $E$ be a COGEA and let $(\eta\sb{e})\sb{e\in E}$ be the
hull system on $E$ determined by the complete boolean algebra
and CD set $\Sigma\sb{\sim}(E)$.  Then $c\in E$ is
\emph{$\eta$-invariant} iff $\eta\sb{c}(E)=E[0,c]$.\footnote{A
different but equivalent definition of $\eta$-invariance is given
in \cite[Definition 7.1]{ExoCen}---see \cite[Lemma 7.3]{ExoCen}.}
We denote by $\Gamma\sb{\eta}(E)$ the set of all $\eta$-invariant elements
in the COGEA $E$.
\end{definition}

\begin{theorem} [{\cite[Theorems 7.5 and 7.7]{ExoCen},
\cite[Theorem 7.3]{HDTD} }] \label{th:COGEAetainv}
Suppose that $E$ is a COGEA and let $(c\sb{i})\sb{i\in I}\subseteq\Gamma
\sb{\eta}(E)$. Then{\rm:} {\rm(i)} $\Gamma\sb{\eta}(E)$ is a sublattice of
the center $\Gamma(E)$ of $E$, and as such, it is a generalized boolean
algebra. {\rm(ii)} If $I$ is nonempty, then the infimum $c:=\bigwedge\sb{i
\in I}c\sb{i}$ exists in $E$, $c\in\Gamma\sb{\eta}(E)$, and $\eta
\sb{c}=\bigwedge\sb{i\in I}\eta\sb{c\sb{i}}$ in the boolean algebra
$\GEX(E)$. {\rm(iii)} If $(c\sb{i})\sb{i\in I}$ is bounded above in $E$,
then the supremum $s :=\bigvee\sb{i\in I}c\sb{i}$ exists in $E$,
$s\in\Gamma\sb{\eta}(E)$, and $\eta\sb{s}=\bigvee\sb{i\in I}
\eta\sb{c\sb{i}}$ in the boolean algebra $\GEX(E)$.
\end{theorem}

The next theorem, which is the main theorem of this section, shows that,
if $E$ is a Dedekind orthocomplete COGEA, then
\[
\Gamma\sb{\sim}(E)=\Gamma\sb{\eta}(E)=\{c\in\Gamma(E):\pi\sb{c}
 \in\Sigma\sb{\sim}(E)\}\subseteq\Gamma(E).
\]

\begin{theorem} \label{th:maintheorem}
Suppose that $E$ is both centrally orthocomplete {\rm(}a COGEA{\rm)} and
Dedekind orthocomplete, let $(\eta\sb{e})\sb{e\in E}$ be the hull system
determined by the HD set $\Sigma\sb{\sim}(E)$, and let $c\in E$. Then, the
following conditions are mutually equivalent{\rm:}
\begin{enumerate}
\item $c\in\Gamma(E)$ and $\pi\sb{c}\in\Sigma\sb{\sim}(E)$ {\rm(}i.e.,
 $\pi\sb{c}$ splits $\sim${\rm)}.
\item $c\in\Gamma\sb{\eta}(E)$ {\rm(}i.e., $c$ is $\eta$-invariant{\rm)}.
\item $c\in\Gamma\sb{\sim}(E)$ {\rm(}i.e., $c$ is invariant{\rm)}.
\item $c$ is principal and $E[0,c]$ is hereditary.
\item $c$ is principal and, for all $e\in E$, $e\sim c\Rightarrow
 e\leq c$.
\item $c$ is principal and $\{f\in E:f\perp c\}$ is hereditary.
\end{enumerate}
\end{theorem}

\begin{proof}
Assume (i). Then, as $\pi\sb{c}c=c$ and $\eta\sb{c}$ is the smallest
mapping $\pi$ in $\Sigma\sb{\sim}(E)$ such that $\pi c=c$, we have
$\eta\sb{c}\leq\pi\sb{c}$; hence $\eta\sb{c}(E)\subseteq\pi\sb{c}(E)
=E[0,c]$. Also, as $c\in\eta\sb{c}(E)$ and $\eta\sb{c}(E)$ is an ideal
in $E$, we have $E[0,c]\subseteq\eta\sb{c}(E)$; hence $\eta\sb{c}(E)
=E[0,c]$, and it follows that $c\in\Gamma\sb{\eta}(E)$. Thus (i)
$\Rightarrow$ (ii), the converse is obvious, and we have (i)
$\Leftrightarrow$ (ii). By Lemma \ref{lm:hereditary2} and the fact
that principal elements are sharp, it follows that (iii) $\Leftrightarrow$
(iv) $\Leftrightarrow$ (v) $\Leftrightarrow$ (vi). Finally, (i)
$\Leftrightarrow$ (iv) by Lemma \ref{lm:hereditary4}.
\end{proof}

\section{A Dimension Generalized Effect Algebra} \label{sc:DGEA}

\noindent\emph{The assumption that $\sim$ is an SK-congruence on $E$
remains in force.}

\begin{definition} \label{df:DER}
The SK-congruence $\sim$ is called a \emph{dimension equivalence
relation} (DER) on $E$ iff, in addition to SK1--SK4, it satisfies the
following condition:
\begin{enumerate}
\item [(SK4a$'$)] \emph{For all $e,f\in E$, if $e$ is unrelated to
 $f$ then there exists $\pi\in\Sigma\sb{\sim}(E)$ such that $e\in
\pi(E)$ and $f\in\pi\,'(E)$.}
\end{enumerate}
\end{definition}

If $\pi\in\GEX(E)$, $e\in\pi(E)$, and $f\in\pi\,'(E)$, then $e\perp f$;
hence, since condition SK4a is equivalent to the requirement that unrelated
elements in $E$ are orthogonal, it follows that SK4a$'$ is formally
stronger than SK4a. We note that SK4a$'$ is the converse of part (iv)
of Lemma \ref{lm:piinSigmasim}.

\begin{example} \label{ex:etasimDER}
In view of Theorem \ref{th:etasimSK}, \emph{if $E$ is orthogonally
ordered and $(\eta\sb{e})\sb{e\in E}$ is a divisible hull system on
$E$, then $\etasim$ is a DER on $E$.}
\end{example}

As a consequence of Theorem \ref{th:SigmasimProps} (vi), we have the
following.

\begin{lemma} \label{lm:unrelated&DER}
Let $E$ be a COGEA and let $(\eta\sb{e})\sb{e\in E}$ be the hull system
determined by $\Sigma\sb{\sim}(E)$. Then $\sim$ is a DER iff, for all $e,f
\in E$, if $e$ is unrelated to $f$, then $\eta\sb{e}\wedge\eta\sb{f}=0$.
\end{lemma}

Using Lemma \ref{lm:unrelated&DER} and \cite[Remarks 7.6]{ExoCen}, it
is not difficult to show that, if $E$ is an orthocomplete (hence, centrally
orthocomplete) EA, then condition SK4a$'$ in Definition \ref{df:DER} is
equivalent to condition SK4$'$ in \cite[Definition 7.14]{HandD}.

\begin{definition} \label{df:DGEA}
A \emph{dimension generalized effect algebra} (DGEA) is a centrally
orthocomplete (i.e., a COGEA) and Dedekind orthocomplete GEA equipped
with a specified dimension equivalence relation (DER).
\end{definition}

\begin{example}
Let $H$ be an infinite dimensional separable complex Hilbert space, let
${\mathcal C}\sp{+}(H)$ be the set of all positive semi-definite compact
operators on $H$, and let $\tau$ be a faithful normal trace on the set of
all bounded operators on $H$. Then with $A\oplus B:=A+B$ (the usual
operator sum), ${\mathcal C}\sp{+}(H)$ is a GEA \cite[Theorem 3.1]{Pola}.
Define $\sim\sb{\tau}$ on ${\mathcal C}\sp{+}(H)$ as follows: $A\sim
\sb{\tau}B$ iff either $\tau(A)$ and $\tau(B)$ are both infinite, or
$\tau(A)$ and $\tau(B)$ are both finite with $\tau(A)=\tau(B)$. Then
${\mathcal C}\sp{+}(H)$ is a DGEA with $\sim\sb{\tau}$ as the DER.
\end{example}

\begin{assumptions} \label{as:DGEA}
In what follows, we assume that the Dedekind orthocomplete COGEA
$E$ is a DGEA with DER $\sim$ and that $(\eta\sb{e})\sb{e\in E}$
is the hull system determined by the complete boolean algebra and
HD set $\Sigma\sb{\sim}(E)$.
\end{assumptions}

As per the next theorem, condition SK4a$'$ enables us to use the
hull system $(\eta\sb{e})\sb{e\in E}$ to determine whether or not
two elements of $E$ are related.

\begin{theorem}\label{th:UnrelatedConditions}
If $e,f\in E$, then the following conditions are mutually equivalent{\rm:}
{\rm (i)} $e$ is unrelated to $f$. {\rm (ii)} There exists $\pi\in\Sigma
\sb{\sim}(E)$ such that $e\in\pi(E)$ and $f\in\pi\,'(E)$. {\rm (iii)} $\eta
\sb{e}\wedge\eta\sb{f}=0$, i.e., $\eta\sb{e}\leq(\eta\sb{f})'$. {\rm (iv)}
$\eta\sb{e}f=0$. {\rm(v)} $\eta\sb{f}e$ is unrelated to $\eta\sb{e}f$.
\end{theorem}

\begin{proof}
By SK4a$'$ and part (iv) of Lemma \ref{lm:piinSigmasim}, we have (i)
$\Leftrightarrow$ (ii), and (ii) $\Leftrightarrow$ (iii) by Theorem
\ref{th:SigmasimProps} (vi).  That (iii) $\Leftrightarrow$ (iv) follows
from the facts that $\eta\sb{\eta\sb{e}f}=\eta\sb{e}\wedge\eta\sb{f}$
and $\eta\sb{p}=0\Leftrightarrow p=0$ (\cite[Definition 7.1 and Theorem
7.2 (i)]{ExoCen}. We have $\eta\sb{\eta\sb{f}e}\wedge\eta\sb{\eta\sb{e}f}=
\eta\sb{f}\wedge\eta\sb{e}\wedge\eta\sb{e}\wedge\eta\sb{f}=\eta\sb{e}
\wedge\eta\sb{f}$, so (v) $\Leftrightarrow$ (i) follows from the
equivalence of (i) and (iii).
\end{proof}

\begin{corollary} \label{co:desc/unrel}
If $e\in E$, then $E$ decomposes as $E=\eta\sb{e}(E)\oplus(\eta\sb{e})'(E)$,
the elements in the hereditary ideal $\eta\sb{e}(E)$ are precisely the
descendents of $e$, and the elements in the hereditary ideal $(\eta\sb{e})'
(E)$ are precisely the elements of $E$ that are unrelated to $e$.
\end{corollary}

\begin{proof}
As $\eta\sb{e},(\eta\sb{e})'\in\Sigma\sb{\sim}(E)\subseteq\GEX(E)$, both
$\eta\sb{e}(E)$ and $(\eta\sb{e})'(E)$ are hereditary ideals in $E$. If
$d\in\eta\sb{e}(E)$ and $0\not=f\leq d$, then $f\in\eta\sb{e}(E)$, so $0
\not=\eta\sb{e}f$, whence $f$ is related to $e$ by Theorem
\ref{th:UnrelatedConditions}, so $d$ is a descendent of $e$. Conversely,
suppose that $d$ is a descendent of $e$. Then $f:=(\eta\sb{e})'d\leq d$
with $\eta\sb{e}f=0$, whence $f$ is unrelated to $e$ by Theorem
\ref{th:UnrelatedConditions} again, so $f=0$, and therefore $d=\eta\sb{e}d
\in\eta\sb{e}(E)$.

Since $f\in(\eta\sb{e})'(E)$ iff $\eta\sb{e}f=0$, it follows from
Theorem \ref{th:UnrelatedConditions} that $(\eta\sb{e})'(E)$ is
precisely the set of all elements of $E$ that are unrelated to $e$.
\end{proof}

\begin{corollary} \label{co:relatedtoorthosum}
Let $(e\sb{i})\sb{i\in I}$ be an orthosummable family in $E$ with
$e:=\oplus\sb{i\in I}e\sb{i}$. If $f\in E$ and $f$ is unrelated
to $e\sb{i}$ for all $i\in I$, then $f$ is unrelated to $e$.
\end{corollary}

\begin{proof}
Assume the hypotheses. Then by Theorem \ref{th:UnrelatedConditions} (iii),
$\eta\sb{e\sb{i}}\leq(\eta\sb{f})'$ for all $i\in I$. Thus, by
\cite[Theorem 7.4 (ix)]{HDTD} $\eta\sb{e}=\bigvee\eta\sb{e\sb{i}}
\leq(\eta\sb{f})'$, and from Theorem \ref{th:UnrelatedConditions} (iii)
again, we infer that $e$ is unrelated to $f$.
\end{proof}

\begin{theorem} [{Cf. \cite[Theorem 7.27]{HandD}}]
\label{th:SupHeridtaryIdeal}
Let $H$ be a hereditary ideal that is bounded above in
$E$. Then {\rm(i)} $c:=\bigvee H$ exists in $E$ and $c$
is the orthosum of an orthogonal family in $H$. {\rm(ii)}
$\eta\sb{c}=\bigvee\sb{h\in H}\eta\sb{h}$. {\rm(iii)} $c$
is sharp and $E[0,c]$ is hereditary. {\rm(iv)} If $E$ is
directed or orthogonally ordered, then $c\in\Gamma\sb{\sim}(E)
=\Gamma\sb{\eta}(E)\subseteq\Gamma(E)$.
\end{theorem}

\begin{proof}
Assume the hypotheses and choose a maximal orthogonal family $(h\sb{i})
\sb{i\in I}$ in $H$. As $H$ is an ideal, the finite partial orthosums
of $(h\sb{i})\sb{i\in I}$ belong to $H$, hence they are bounded
above in $E$, and it follows from Dedekind orthocompleteness that $c:=
\oplus\sb{i\in I}h\sb{i}$ exists in $E$. By the maximality of $(h\sb{i})
\sb{i\in I}$, no nonzero element of $H$ can be orthogonal to $c$. We
claim that $c$ is an upper bound for $H$, for suppose there exists $h
\in H$ with $h\not\leq c$. Then by SK4b, there exist nonzero $h\sb{1},
h\sb{2}\in E$ such that $h\geq h\sb{1}\sim h\sb{2}\perp c$. Then $h
\sb{2}\subsim h\in H$, so $0\not=h\sb{2}\in H$ with $h\sb{2}\perp c$,
contradicting the maximality of $(h\sb{i})\sb{i\in I}$. Thus, $c$ is
both an upper bound for $H$ and the supremum of a subset of $H$ (namely,
the set of all finite partial orthosums of $(h\sb{i})\sb{i\in I}$),
whence it is the supremum of $H$, proving (i). Therefore, according to \cite
[Theorem 7.4 (ix)]{HDTD}, we have $\eta\sb{c}=\bigvee\sb{h\in H}\eta\sb{h}$,
proving (ii).

To prove that $c$ is sharp, suppose that there exists $p\in E$ with
$0\not=p\leq c$ and $p\perp c$; hence $p$ is related to $c$ and $p
\perp c$. As $p$ is related to $c$, Corollary \ref{co:relatedtoorthosum}
implies that $p$ is related to $h\sb{j}$ for some $j\in I$. Thus, there
exist nonzero $p\sb{1},q\in E$ with $p\geq p\sb{1}\sim q\leq h\sb{j}$.
Therefore, $p\sb{1}\subsim h\sb{j}\in H$, and it follows that
$0\not=p\sb{1}\in H$. But $p\sb{1}\leq p\perp c$, so $p\sb{1}\perp
c$, contradicting the maximality of $(h\sb{i})\sb{i\in I}$.
Consequently, $c$ is sharp.

To prove that $E[0,c]$ is hereditary, suppose $f\in E$ with
$f\subsim c$, but $f\not\leq c$. Then by SK4b, there are nonzero
$f\sb{1}, d\in E$ with $f\geq f\sb{1}\sim d\perp c$. Suppose that
$d$ is related to $h\sb{j}$ for some $j\in I$, i.e., there exist
nonzero $a,b\in E$ with $d\geq a\sim b\leq h\sb{j}$. Then $a
\subsim h\sb{j}\in H$, so $0\not=a\in H$ with $a\leq d\perp c$,
so $a\perp c$, contradicting the maximality of $(h\sb{i})\sb{i
\in I}$. Therefore, $d$ is unrelated to $h\sb{i}$ for all $i\in I$,
and it follows from Corollary \ref{co:relatedtoorthosum} that $d$
is unrelated to $c=\oplus\sb{i\in I}h\sb{i}$. But $0\not=d\sim f
\sb{1}\leq f\subsim c$, so $f\sb{1}$ is unrelated to $c$ and
$0\not=f\sb{1}\subsim c$, which is a contradiction. This completes
the proof of (iii). Part (iv) follows immediately from (iii) and
Lemma \ref{lm:dir&invar} (iii).
\end{proof}

\begin{theorem} [{General Comparability Theorem for $\subsim$}]
\label{th:Compar}
If $e,f\in E$, there exists $d\in E$ such that $\eta\sb{d} e
\subsim \eta\sb{d}f$ and $(\eta\sb{d})'f\subsim(\eta\sb{d})'e$.
\end{theorem}

\begin{proof}
By Theorem \ref{th:decompospq}, there are orthogonal decompositions
$e=e\sb{1}\oplus e\sb{2}$ and $f=f\sb{1}\oplus f\sb{2}$ such that
$e\sb{1}\sim f\sb{1}$ and $e\sb{2}$ is unrelated to $f\sb{2}$. Put
$d:=f\sb{2}$. Then $\eta\sb{d}f\sb{2}=f\sb{2}$, whence $(\eta\sb
{d})'f\sb{2}=0$, and by Theorem \ref{th:UnrelatedConditions} (iv),
$\eta\sb{d}e\sb{2}=0$. Also, by Theorem \ref{th:pi&sim} (iii) with
$\pi:=\eta\sb{d}$, we have $\eta\sb{d}e\sb{1}\sim\eta\sb{d}f\sb{1}$
and $(\eta\sb{d})'f\sb{1}\sim(\eta\sb{d})'e\sb{1}$, whereupon $\eta
\sb{d}e=\eta\sb{d}e\sb{1}\oplus\eta\sb{d}e\sb{2}=\eta\sb{d}e\sb{1}
\sim\eta\sb{d}f\sb{1}\leq\eta\sb{d}f$ and $(\eta\sb{d})'f=(\eta
\sb{d})'f\sb{1}\oplus(\eta\sb{d})'f\sb{2}=(\eta\sb{d})'f\sb{1}
\sim(\eta\sb{d})'e\sb{1}\leq(\eta\sb{d})'e$. Therefore, $\eta\sb{d}e
\subsim\eta\sb{d}f$ and $(\eta\sb{d})'f\subsim(\eta\sb{d})'e$.
\end{proof}

\begin{definition} \label{df:factor}
The DGEA $E$ is a \emph{factor} iff $\Sigma\sb{\sim}(E)=\{0,1\}$.
\end{definition}

\begin{theorem}\label{th:factor}
The following conditions are mutually equivalent: {\rm(i)} $E$ is a factor.
{\rm(ii)} If $0\not=d\in E$, then $\eta\sb{d}=1$. {\rm(iii)} For all $e,f\in E$,
$e\subsim f$, or $f\subsim e$. {\rm(iv)} Any two nonzero elements of $E$ are
related.
\end{theorem}

\begin{proof}
If (i) holds and $d\in E$, then since $\eta\sb{d}\in\Sigma\sb{\sim}(E)=
\{0,1\}$ we have $\eta\sb{d}=0$ or $\eta\sb{d}=1$; hence, since $\eta
\sb{d}=0$ iff $d=0$, (ii) follows. Conversely, suppose that (ii) holds
and $\pi\in\Sigma\sb{\sim}(E)$ with $\pi\not=0$. Then there exists $d\in
\pi(E)$ with $0\not=d=\pi d$ and, since $\eta\sb{d}$ is the smallest mapping
$\xi\in\Sigma\sb{\sim}(E)$ such that $\xi d=d$, it follows that $1=\eta\sb{d}
\leq\pi$, so $\pi=1$, proving (i). Therefore (i) $\Leftrightarrow$ (ii).

If $E$ is a factor and $d\in E$, then $\eta\sb{d}\in\Sigma\sb{\sim}
(E)$, so $\eta\sb{d}=0$ or $\eta\sb{d}=1$, whence (i) $\Rightarrow$
(iii) follows immediately from Theorem \ref{th:Compar}.

If $e\not=0$ and $e\subsim f$, then $e\geq e\sim f\sb{1}\leq f$ with
$e,f\sb{1}\not=0$, whence $e$ is related to $f$. Likewise, if $f\not=0$
and $f\subsim e$, then again $e$ is related to $f$. Therefore, (iii)
$\Rightarrow$ (iv).

To prove (iv) $\Rightarrow$ (i), assume (iv), but suppose that there
exists $\pi\in\Sigma\sb{\sim}(E)$ with $\pi\not=0$ and $\pi\not=1$.
Then $\pi\,'\not=0$, so there exist $0\not=e\in\pi(E)$ and $0\not=
f\in\pi\,'(E)$, whence $e$ is unrelated to $f$ by Lemma \ref
{lm:piinSigmasim} (iv), contradicting (iv).
\end{proof}

\begin{remarks} \label{rm:factor}
If $E$ is a factor, then by Theorem \ref{th:factor}, $\eta\sb{e}=1
\Leftrightarrow e\not=0$ for all $e\in E$, i.e., $(\eta\sb{e})\sb{e
\in E}$ is the indiscrete hull system on $E$ \cite[Example 6.10]{HandD}.
Assume that $E$ is a factor. Then, as is easily seen, the nonzero
$\eta$-monads are precisely the atoms in $E$. Moreover, if $0\not=e
\in E$ and $e$ is not an atom, then there are nonzero elements
$e\sb{1}, e\sb{2}\in E$ with $e=e\sb{1}\oplus e\sb{2}$, and we have
$\eta\sb{e\sb{1}}=\eta\sb{e\sb{2}}=1$. Therefore, \emph{every nonzero
element in a factor is either an atom or an $\eta$-dyad, but not
both.} \hspace{\fill} $\square$
\end{remarks}

\section{Hereditary $\eta$STD Sets in a DGEA}

\noindent\emph{Standing Assumptions \ref{as:DGEA} remain in force.}

\begin{definition} \label{df:faithful}
An element $p\in E$ is \emph{faithful} iff $\eta\sb{p}=1$
(cf. \cite[Definition 10.2 (4)]{HDTD}).
\end{definition}

\begin{theorem} \label{th:hstar}
Let $H\subseteq E$ be a hereditary $\eta$STD set. Then{\rm:}
\begin{enumerate}
\item If $h\in H$, $e\in E$, and $\eta\sb{h}e\not=0$, there exists
 $h\sb{1}\in H$ such that $0\not=h\sb{1}\leq\eta\sb{h}e$.
\item There exists $h\sp{\ast}\in H$ such that $\eta\sb{h\sp{\ast}}$ is the
 largest mapping in $\{\eta\sb{h}:h\in H\}$.
\item $\eta\sb{h\sp{\ast}}=\bigvee\sb{h\in H}\eta\sb{h}$ and $H\subseteq
 \eta\sb{h\sp{\ast}}(E)$.
\item There exists a faithful element in $H$ iff $\eta\sb{h\sp{\ast}}=1$.
\item $H$ is orthodense in the hereditary direct summand $\eta\sb{h\sp
 {\ast}}(E)$ of $E$.
\item If $\pi\in\GEX(E)$, then $H\cap\pi(E)=\{\pi h:h\in H\}$.
\end{enumerate}
\end{theorem}

\begin{proof}
(i) Assume the hypotheses of (i). Then by Theorem \ref
{th:UnrelatedConditions}, $\eta\sb{h}e$ and $\eta\sb{e}h$ are related,
whence there are nonzero $h\sb{1}, h\sb{2}\in E$ with $\eta\sb{h}e
\geq h\sb{1}\sim h\sb{2}\leq\eta\sb{e}h\leq h$. Thus, $h\sb{1}\subsim h
\in H$, and since $H$ is hereditary, $h\sb{1}\in H$, proving (i).
Part (ii) follows from Theorem \ref{th:tStar}, (iii) follows from (ii)
and \cite[Theorem 7.5]{HDTD}, and (iv) is an obvious consequence of
(ii). Part (v) follows from (i), (iii), and \cite[Theorem 7.5]{HDTD}.
Finally,  to prove (vi), let $\pi\in\GEX(E)$ and $h\in H$. If $h\in
\pi E$, then $h=\pi h\in\{\pi h:h\in H\}$; conversely, since $\pi h
\leq h$ and $H$ is an order ideal, it follows that $\pi h\in H\cap
\pi(E)$.
\end{proof}

\begin{lemma} \label{lm:HcappiE=0}
Let $\pi\in\Sigma\sb{\sim}(E)$, let $H\subseteq E$ be a hereditary
$\eta$STD set, and choose $h\sp{\ast}\in H$ such that $\eta\sb{h
\sp{\ast}}=\bigvee\sb{h\in H}\eta\sb{h}$ {\rm(}Theorem \ref
{th:hstar}{\rm)}.  Then the following conditions are mutually
equivalent{\rm:} {\rm(i)} $H\cap\pi(E)=\{0\}$. {\rm(ii)} $H\subseteq
\pi\,'(E)$. {\rm(iii)} $\pi\leq(\eta\sb{h\sp{\ast}})'$, i.e.,
$\pi\wedge\eta\sb{h\sp{\ast}}=0$. {\rm(iv)} $h\in H\Rightarrow\pi
\wedge\eta\sb{h}=0$.
\end{lemma}

\begin{proof}
(i) $\Rightarrow$ (ii). Assume (i) and let $h\in H$. Then by
Theorem \ref{th:hstar} (v), $\pi h\in H\cap\pi(E)=\{0\}$,
whence $h\in\pi\,'(E)$.

(ii) $\Rightarrow$ (iii). Assume (ii). Then $\pi\,'h\sp{\ast}=
h\sp{\ast}$. Thus as $\pi\in\Sigma\sb{\sim}(E)$, we have $\pi\,'
\in\Sigma\sb{\sim}(E)$, and it follows from Theorem \ref
{th:SigmasimProps} (iii) that $\eta\sb{h\sp{\ast}}\leq\pi\,'$,
so (iii) holds.

(iii) $\Rightarrow$ (iv). If $h\in H$, then $\eta\sb{h}\leq\eta
\sb{h\sp{\ast}}$, so $\pi\wedge\eta\sb{h\sp{\ast}}=0\Rightarrow
\pi\wedge\eta\sb{h}=0$.

(iv) $\Rightarrow$ (i). Assume (iv) and let $h\in H\cap\pi(E)$.
Then $h=\eta\sb{h}h$ and $h=\pi h$, so $h=\eta\sb{h}h\wedge\pi h
=(\eta\sb{h}\wedge\pi)h=0$, whence (i) holds.
\end{proof}

If $\pi\in\GEX(E)$ and $\phi$ is a function defined on $E$, then
$\phi|\sb{\pi}$ denotes the restriction of $\phi$ to the direct
summand $\pi(E)$ of $E$ and $\sim\sb{\pi}$ denotes the restriction
of $\sim$ to $\pi(E)$. If $\xi\in\GEX(E)$, we write $\xi|\sb{\pi}$,
for short, as $\xi\sb{\pi}$.

\begin{theorem} \label{th:piEasDGEA}
Suppose that $\pi\in\Sigma\sb{\sim}(E)$. Then{\rm:} {\rm(i)}
$\GEX(\pi(E))=\{\xi\sb{\pi}:\xi\in\GEX(E)\}$ and $\xi\mapsto\xi
\sb{\pi}$ is a surjective boolean homomorphism of $\GEX(E)$ onto
$\GEX(\pi(E))$. {\rm(ii)} $(\eta\sb{p}|\sb{\pi})\sb{p\in\pi(E)}$ is
a hull system on $\pi(E)$. {\rm(iii)} The hereditary direct summand
$\pi(E)$ of $E$ is a DGEA with $\sim\sb{\pi}$ as the DER and
$(\eta\sb{p}|\sb{\pi})\sb{p\in\pi(E)}$ is the corresponding hull
system determined by $\Sigma\sb{\sim\sb{\pi}}(\pi(E))=\{\xi\sb{\pi}:
\xi\in\Sigma\sb{\sim}(E)\}$. {\rm(v)} $\Gamma\sb{\sim\sb{\pi}}(\pi(E))
=\Gamma\sb{\sim}(E)\cap\pi(E)=\{\pi c:c\in\Gamma\sb{\sim}(E)\}$.
{\rm(iv)} If $H\subseteq E$ is a hereditary $\eta$STD set, then $H
\cap\pi(E)=\{\pi h:h\in H\}$ is a hereditary {\rm(}with respect to
$\sim\sb{\pi}${\rm)} $\eta|\sb{\pi}$STD set in the DGEA $\pi(E)$.
\end{theorem}

\begin{proof}
Part (i) follows from \cite[Theorem 4.13 (i)--(iii)]{CenGEA} and
part (ii) is \cite[Theorem 9.5]{HDTD}. We omit the remainder of the
proof as it is straightforward.
\end{proof}

In what follows, if $\pi\in\Sigma\sb{\sim}(E)$, \emph{we understand
that the hereditary direct summand $\pi(E)$ is organized into a
DGEA as per Theorem \ref{th:piEasDGEA} {\rm(}iii{\rm)}}. Also we
recall that $\Theta\sb{\eta}(E):=\{\eta\sb{e}:e\in E\}\subseteq
\Sigma\sb{\sim}(E)$.

\begin{lemma}  \label{lm:hsharp}
Let $\pi\in\Theta\sb{\eta}(E)$, let $H\subseteq E$ be a hereditary
$\eta$STD set, choose $h\sp{\ast}\in H$ such that $\eta\sb{h\sp{\ast}}
=\bigvee\sb{h\in H}\eta\sb{h}$ {\rm(Theorem \ref{th:hstar})}, and put
$h\sp{\#}:=\pi h\sp{\ast}$. Then{\rm:} {\rm(i)} $h\sp{\#}\in H\cap
\pi(E)$. {\rm(ii)} $\eta\sb{h\sp{\#}}=\eta\sb{h\sp{\ast}}\circ\pi=
\eta\sb{h\sp{\ast}}\wedge\pi$. {\rm(iii)} $\eta\sb{h\sp{\#}}|\sb{\pi}$
is the largest mapping in $\{\eta\sb{h}|\sb{\pi}:h\in H\cap\pi(E)\}$.
{\rm(iv)} $H\cap\pi(E)$ is orthodense in $\eta\sb{h\sp{\ast}}(\pi(E))
=(\eta\sb{h\sp{\ast}}\wedge\pi)(E)=\eta\sb{h\sp{\#}}(E)$.
\end{lemma}

\begin{proof}
We have $h\sp{\#}=\pi h\sp{\ast}\leq h\sp{\ast}\in H$, and as $H$ is
an order ideal, (i) follows. Also, since $\pi\in\Theta\sb{\eta}(E)$,
there exists $f\in E$ with $\pi=\eta\sb{f}$, and therefore
$\eta\sb{h\sp{\#}}=\eta\sb{\pi h\sp{\ast}}=\eta\sb{\eta\sb{f}h
\sp{\ast}}=\eta\sb{h\sp{\ast}}\circ\eta\sb{f}=\eta\sb{h\sp{\ast}}
\circ\pi=\eta\sb{h\sp{\ast}}\wedge\pi$, proving (ii). Let $h\in H
\cap\pi(E)$ and $e\in E$. Then $\eta\sb{h}\leq\eta\sb{h\sp{\ast}}$,
so by (ii), $\eta\sb{h}(\pi e)\leq\eta\sb{h\sp{\ast}}(\pi e)=(\eta
\sb{h\sp{\ast}}\circ\pi)(\pi e)=\eta\sb{h\sp{\#}}(\pi e)$, and it
follows that $\eta\sb{h}|\sb{\pi}\leq\eta\sb{h\sp{\#}}|\sb{\pi}$,
and (iii) is proved. Part (iv) follows from (iii), Theorem
\ref{th:piEasDGEA} (iv), and Theorem \ref{th:hstar} with $E$ replaced
by $\pi(E)$, $H$ replaced by $H\cap\pi(E)$, and $\eta\sb{h\sp{\ast}}$
replaced by $\eta\sb{h\sp{\#}}|\sb{\pi}$.
\end{proof}

\begin{lemma} \label{lm:pfaithful}
If $\pi\in\Sigma\sb{\sim}(E)$ and $p\in\pi(E)$, then the following
conditions are mutually equivalent{\rm:} {\rm(i)} $p$ is faithful
in the DGEA $\pi(E)$. {\rm(ii)} $\pi\leq\eta\sb{p}$. {\rm(iii)}
$\pi=\eta\sb{p}$. Thus, if any {\rm(}hence all{\rm)} of these
conditions hold, then $\pi\in\Theta\sb{\eta}(E)$.
\end{lemma}

\begin{proof}
Assume the hypotheses of the lemma. Then $p$ is faithful in
$\pi(E)$ iff $\eta\sb{p}|\sb{\pi}$ is the identity mapping on
$\pi(E)$ iff $(\eta\sb{p}\wedge\pi)e=\eta\sb{p}(\pi e)=\pi e$
for all $e\in E$ iff $\eta\sb{p}\wedge\pi=\pi$ iff $\pi\leq\eta
\sb{p}$, proving that (i) $\Leftrightarrow$ (ii). Since $p\in
\pi(E)$, we have $\pi p=p$, whence $\eta\sb{p}\leq\pi$ by
Theorem \ref{th:SigmasimProps} (iii), and it follows that (ii)
$\Leftrightarrow$ (iii).
\end{proof}

\begin{theorem} \label{th:faithinpi}
Let $\pi\in\Theta\sb{\eta}(E)$ and let $H\subseteq E$ be a hereditary
$\eta$STD set. Then with the notation of Lemma \ref{lm:hsharp}, the
following conditions are mutually equivalent{\rm:} {\rm(i)} There
exists $h\in H$ such that $\eta\sb{h}=\pi$. {\rm(ii)} $\pi\leq\eta
\sb{h\sp{\ast}}$. {\rm(iii)} $h\sp{\#}$ is faithful in $\pi(E)$. {\rm(iv)}
There exists $h\in H\cap\pi(E)$ such that $h$ is faithful in $\pi(E)$.
Moreover, if any {\rm(}hence all{\rm)} of the conditions {\rm(i)--(iv)}
hold, then $H\cap\pi(E)$ is orthodense in $\pi(E)$.
\end{theorem}

\begin{proof}
That (i) $\Leftrightarrow$ (ii) follows from \cite[Definition 12.6 (2)
and Theorem 12.8 (iii)]{HDTD}. Also, by \cite[Definition 12.6 (2) and
Theorem 12.9]{HDTD}, we have (i) $\Leftrightarrow$ (iii). Obviously,
(iii) $\Rightarrow$ (iv). If (iv) holds, then by Lemma \ref{lm:pfaithful},
$\pi\leq\eta\sb{h}$, and since $\eta\sb{h}\leq\eta\sb{h\sp{\ast}}$, we
have (ii); hence (iv) $\Rightarrow$ (ii), so (i)--(iv) are mutually
equivalent. Finally, if (ii) holds, then $\pi\leq\eta\sb{h\sp{\ast}}$,
so by Lemma \ref{lm:hsharp} (iv), $H\cap\pi(E)$ is orthodense in
$(\eta\sb{h\sp{\ast}}\wedge\pi)(E)=\pi(E)$.
\end{proof}

\section{Simple and Finite Elements} \label{sc:SandF} 

\noindent\emph{Standing Assumptions \ref{as:DGEA} remain in force.}

\begin{definition}[{Cf. \cite[p. 6]{Loom}}] \label{df:simple&finite}
Let $k,f\in E$. Then:
\begin{enumerate}
\item[(1)] $k$ is \emph{simple} iff every subelement $e$ of $k$ is
 unrelated to its orthosupplement $e\sp{\perp\sb{k}}=k\ominus e$
 in the EA $E[0,k]$. In what follows, we denote by $K$ the set of
all simple elements in $E$ and we define $\eta\sb{K}:=\bigvee
\sb{k\in K}\eta\sb{k}$.
\item[(2)] $f$ is \emph{finite} iff, for every subelement $e$ of
 $f$, $e\sim f\Rightarrow e=f$. If $p\in E$ is not finite, it is
 \emph{infinite}. In what follows, $F$ is the set of all finite
 elements in $E$, $\eta\sb{F}:=\bigvee\sb{f\in F}\eta\sb{f}$, and
 ${\widetilde F}:=F\cap\Gamma\sb{\eta}(E)=F\cap
\Gamma\sb{\sim}(E)$, i.e., ${\widetilde F}$ is the set of all finite
 invariant elements in $E$.
\end{enumerate}
\end{definition}

\begin{lemma} \label{lm:simp=mon}
Let $k\in E$. Then the following conditions are mutually equivalent{\rm:}
{\rm(i)} $k$ is simple. {\rm(ii)} If $k=k\sb{1}\oplus k\sb{2}$, then
$\eta\sb{k\sb{1}}\wedge\eta\sb{k\sb{2}}=0$. {\rm(iii)} If $e,f\in E[0,k]$
and $e\perp\sb{k}f$, then $\eta\sb{e}\wedge\eta\sb{f}=0$. {\rm(iv)} If
$e\in E[0,k]$, then $e=\eta\sb{e}k$. {\rm(v)} $k$ is an $\eta$-monad.
\end{lemma}

\begin{proof}
That (i) $\Leftrightarrow$ (ii) follows from Theorem \ref{th:UnrelatedConditions}.
If (ii) holds, and $e,f\in E[0,k]$ with $e\perp\sb{k}f$, then $f\leq k\ominus e$,
whence $\eta\sb{e}\wedge\eta\sb{f}\leq\eta\sb{e}\wedge\eta\sb{(k\ominus e)}=0$,
so (ii) $\Rightarrow$ (iii). Obviously, (iii) $\Rightarrow$ (ii), so (ii)
$\Leftrightarrow$ (iii). That (iii) $\Leftrightarrow$ (iv) $\Leftrightarrow$ (v)
follows from \cite[Theorem 10.9]{HDTD}.
\end{proof}

\begin{lemma} \label{lm:simeta=sim}
Let $q,k\in K$. Then{\rm:} {\rm(i)} $\eta\sb{q}\leq\eta\sb{k}
\Leftrightarrow q\subsim k$. {\rm(ii)} $q\sim\sb{\eta}k\Leftrightarrow
q\sim k$.
\end{lemma}

\begin{proof}
Assume $q,k\in K$. (i) Suppose $\eta\sb{q}\leq\eta\sb{k}$. By Theorem
\ref{th:decompospq}, Theorem \ref{th:pi&sim} (iii), and Theorem \ref
{th:UnrelatedConditions} $q=q\sb{1}\oplus q\sb{2}$ and $k=k\sb{1}\oplus
k\sb{2}$ with $q\sb{1}\sim k\sb{1}$, $\eta\sb{q\sb{1}}=\eta\sb{k\sb{1}}$
and $\eta\sb{q\sb{2}}\wedge\eta\sb{k\sb{2}}=0$. Also, by Lemma
\ref{lm:simp=mon}, $\eta\sb{q\sb{1}}\wedge\eta\sb{q\sb{2}}=0$ and
$\eta\sb{k\sb{1}}\wedge\eta\sb{k\sb{2}}=0$. Moreover, \cite[Theorem 7.4
(ix)]{HDTD} implies that $\eta\sb{q\sb{1}}\vee\eta\sb{q\sb{2}}=\eta
\sb{q}\leq\eta\sb{k}=\eta\sb{k\sb{1}}\vee\eta\sb{k\sb{2}}$. Working in
the boolean algebra $\GEX(E)$, we deduce from these conditions that
$\eta\sb{q\sb{2}}=0$, whence $q\sb{2}=0$, so $q=q\sb{1}\sim k\sb{1}
\leq k$, and we have $q\subsim k$. Conversely, $q\subsim k\Rightarrow
\eta\sb{q}\leq\eta\sb{k}$ by Theorem \ref{th:pi&sim} (i), completing
the proof of (i). Finally, (ii) follows from (i) and Theorem \ref{th:preorder}.
\end{proof}

\begin{lemma} \label{lm:simpisfinite}
Every simple element {\rm(}$\eta$-monad{\rm)} in $E$ is finite, i.e.,
$K\subseteq F$; hence $\eta\sb{K}\leq\eta\sb{F}$.
\end{lemma}

\begin{proof}
Suppose that $k\in K$ and $e\in E[0,k]$ with $e\sim k$. By Lemma
\ref{lm:simp=mon}, $e=\eta\sb{e}k$, and by Theorem \ref{th:pi&sim} (iii),
$\eta\sb{e}=\eta\sb{k}$, whence $e=\eta\sb{k}k=k$, so $k\in F$.
\end{proof}

\begin{lemma} \label{lm:subtraction}
{\rm(i)} $F$ is an order ideal. {\rm(ii) (Subtraction Property)} Let
$e=e\sb{1}\oplus e\sb{2}$ and $f=f\sb{1}\oplus f\sb{2}$ where $e,f\in F$,
$e\sim f$ and $e\sb{1}\sim f\sb{1}$. Then $e\sb{2}\sim f\sb{2}$.
\end{lemma}

\begin{proof}
(i) Let $e\in E$, $e\leq f\in F$, $e=e\sb{1}\oplus e\sb{2}$ with
$e\sim e\sb{1}$. Then by finite additivity, $f\ominus e\sb{2}=
e\sb{1}\oplus(f\ominus e)\sim e\oplus(f\ominus e)=f$; hence
$f\ominus e\sb{2}=f$, therefore $e\sb{2}=0$, and we have $e\sb{1}
=e$. Consequently, $e\in F$, and (i) is proved.

(ii) By Theorem \ref{th:Compar}, there exists $\pi\in\Sigma\sb{\sim}(E)$ with
$\pi e\sb{2}\subsim\pi f\sb{2}$ and $\pi\,'f\sb{2}\subsim\pi\,'e\sb{2}$. As
$\pi e\sb{2}\subsim\pi f\sb{2}$, there exist $s, t\in E$ with $\pi e\sb{2}
\sim s$ and $\pi f\sb{2}=s\oplus t$. As $e\sim f$ and $e\sb{1}\sim f\sb{1}$,
we have $\pi e\sim\pi f$ and $\pi e\sb{1}\sim\pi f\sb{1}$ by Theorem \ref
{th:simCoordinatewise}. Now $\pi f =\pi f\sb{1}\oplus\pi f\sb{2}=\pi f\sb{1}
\oplus s\oplus t$ and by finite additivity $\pi f\sb{1}\oplus s\sim\pi e\sb{1}
\oplus\pi e\sb{2}=\pi e\sim\pi f$. Since $f\in F$ and $\pi f\leq f$, part
(i) implies that $\pi f\in F$; hence $\pi f\sb{1}\oplus s=\pi f$, so $t=0$,
and therefore $\pi f\sb{2}=s\sim\pi e\sb{2}$.  Similarly, $\pi\,'f\sb{2}
=\pi\,'e\sb{2}$, so $e\sb{2}\sim f\sb{2}$ by additivity.
\end{proof}

\begin{theorem} \label{th:KFheretaSTD}
{\rm(i)} $K$ is both a hereditary order ideal and an $\eta$STD set in
 $E$. {\rm(ii)} $F$ is both a hereditary ideal and an $\eta$STD set in
 $E$. {\rm(iii)} ${\widetilde F}$ is an $\eta$TD subset of $E$.
\end{theorem}

\begin{proof}
(i) According to \cite[Theorem 11.13]{HDTD}, the set $K$ of all
simple elements ($\eta$-monads) in $E$ forms an $\eta$STD set.
In particular, $K$ is an order ideal, so to prove that it is
hereditary, it will be sufficient to prove that if $e\sim k\in K$,
then $e\in K$. Thus, suppose that $e\sim k\in K$ and $e=e\sb{1}
\oplus e\sb{2}$. Then by divisibility (SK3d), $k=k\sb{1}\oplus
k\sb{2}$ with $e\sb{1}\sim k\sb{1}$ and $e\sb{2}\sim k\sb{2}$.
By Lemma \ref{lm:simp=mon}, $\eta\sb{k\sb{1}}\wedge\eta\sb{k\sb{2}}
=0$; by Theorem \ref{th:pi&sim} (iii), $\eta\sb{e\sb{1}}=\eta\sb
{k\sb{1}}$ and $\eta\sb{e\sb{2}}=\eta\sb{k\sb{2}}$; so $\eta\sb
{e\sb{1}}\wedge\eta\sb{e\sb{2}}=0$; whence $e\in K$ by
Lemma \ref{lm:simp=mon} again.

(ii) We may apply \cite[Prop. 7.18]{JePu02} (the proof of which
requires Lemma \ref{lm:subtraction} (ii)) to infer that $F$ is
closed under the formation of finite orthosums; hence, in view of
Lemma \ref{lm:subtraction} (i), $F$ is an ideal. Thus, to
prove that $F$ is hereditary, it will be sufficient to show
that if $e\sim f\in F$, then $e\in F$. So assume that $e
=e\sb{1}\oplus e\sb{2}\sim f\in F$ with $e\sb{1}\sim e$. Then
by divisibility, $f=f\sb{1}\oplus f\sb{2}$ with $f\sb{1}\sim
e\sb{1}\sim e\sim f$ and $f\sb{2}\sim e\sb{2}$. Since $f\in F$,
it follows that $f\sb{1}=f$, so $f\sb{2}=0$, whence $e\sb{2}=0$,
and therefore $e\sb{1}=e$. Consequently, $e\in F$, so $F$ is
hereditary.

Let $(f\sb{i})\sb{i\in I}$ be an $\eta$-orthogonal family in $F$,
let $f:=\oplus\sb{i\in I}f\sb{i}=\bigvee\sb{i\in I}f\sb{i}$ (\cite
[Lemma 6.2 (iii)]{ExoCen}), and suppose $e\leq f$ with $e
\sim f$. Let $i\in I$. Then $\eta\sb{f\sb{i}}e\leq\eta\sb{f\sb{i}}
f$, and by \cite[Theorem 7.4 (iii)]{HDTD},
\[
\eta\sb{f\sb{i}}f=\eta\sb{f\sb{i}}(\bigvee\sb{j\in I}f\sb{j})=
 \bigvee\sb{j\in I}\eta\sb{f\sb{i}}f\sb{j}=\bigvee\sb{j\in I}\eta
 \sb{f\sb{i}}(\eta\sb{f\sb{j}}f\sb{j})=\bigvee\sb{j\in I}(\eta
 \sb{f\sb{i}}\wedge\eta\sb{f\sb{j}})f\sb{j}=\eta\sb{f\sb{i}}f
 \sb{i}=f\sb{i},
\]
whence $\eta\sb{f\sb{i}}e\leq f\sb{i}$. Also, by Theorem \ref{th:pi&sim}
(iii), $\eta\sb{f\sb{i}}e\sim\eta\sb{f\sb{i}}f=f\sb{i}$, and since $f
\sb{i}\in F$, it follows that $\eta\sb{f\sb{i}}e=f\sb{i}$. Therefore,
by \cite[Theorem 7.4 (v), (vi), and (ix)]{HDTD}, $e=\eta\sb{f}e=
\bigvee\sb{i\in I}\eta\sb{f\sb{i}}e=\bigvee\sb{i\in I}f\sb{i}=f$;
hence $f\in F$, so $F$ is an $\eta$STD subset of $E$, proving (ii).
According to \cite[Theorem 11.7]{HDTD}, $\Gamma\sb{\eta}(E)$ is an
$\eta$TD subset of $E$, and as the intersection of $\eta$TD subsets
is again $\eta$TD, we have (iii).
\end{proof}

\noindent In view of Definition \ref{df:simple&finite} and Theorem
\ref{th:hstar}, we obtain the following corollary of Theorem
\ref{th:KFheretaSTD}.

\begin{corollary} \label{co:KF}
{\rm(i)} $\eta\sb{K}$ is the largest mapping in $\{\eta\sb{k}:k\in K\}$,
$\eta\sb{K}\in\Theta\sb{\eta}(E)$, $K\subseteq\eta\sb{K}(E)$, $K$ is
orthodense in the hereditary direct summand $\eta\sb{K}(E)$ of $E$,
and $K=\{0\}$ iff $\eta\sb{K}=0$.
{\rm(ii)} $\eta\sb{F}$ is the largest mapping in $\{\eta\sb{f}:f\in F\}$,
$\eta\sb{F}\in\Theta\sb{\eta}(E)$, $F\subseteq\eta\sb{F}(E)$, $F$ is
orthodense in the hereditary direct summand $\eta\sb{F}(E)$ of $E$,
and $F=\{0\}$ iff $\eta\sb{F}=0$.
\end{corollary}

We omit the straightforward proof of the next lemma.

\begin{lemma} \label{lm:KcappiE}
Suppose that $\pi\in\Sigma\sb{\sim}(E)$. Then{\rm:} {\rm(i)}
The set of simple elements in the DGEA $\pi(E)$ is $K\cap\pi(E)=
\{\pi k:k\in K\}$. {\rm(ii)} The set of finite elements in the DGEA
$\pi(E)$ is $F\cap\pi(E)=\{\pi f:f\in F\}$.
\end{lemma}

Combining Theorems \ref{th:faithinpi} and \ref{th:KFheretaSTD}, we
obtain the following result.

\begin{theorem} \label{th:KFfaithinpi}
Let $\pi\in\Theta\sb{\eta}(E)$. Then the following conditions are
mutually equivalent{\rm:} {\rm(i)} There exists $k\in K$ such that
$\pi=\eta\sb{k}$. {\rm(ii)} $\pi\leq\eta\sb{K}$. {\rm(iii)} There
exists $k\in K\cap\pi(E)$ such that $k$ is faithful in $\pi(E)$.
Moreover, if any {\rm(}hence all{\rm)} of the conditions
{\rm(i)--(iii)} hold, then $K\cap\pi(E)$ is orthodense in $\pi(E)$.
Likewise, the following conditions are mutually equivalent{\rm:}
{\rm(iv)} There exists $f\in F$ such that $\pi=\eta\sb{f}$. {\rm(v)}
$\pi\leq\eta\sb{F}$. {\rm(v)} There exists $f\in F\cap\pi(E)$ such
that $f$ is faithful in $\pi(E)$. Moreover, if any {\rm(}hence
all{\rm)} of the conditions {\rm(iv)--(vi)} hold, then $F\cap\pi(E)$
is orthodense in $\pi(E)$.
\end{theorem}

\begin{lemma} \label{lm:etasbtildeF1}
{\rm(i)} There exists an element ${\widetilde f}\in{\widetilde F}$
such that $\eta\sb{\widetilde f}$ is the largest mapping in $\{\eta
\sb{f}:f\in{\widetilde F}\}$. {\rm(ii)} ${\widetilde F}\subseteq
\eta\sb{\widetilde f}(E)=E[0,{\widetilde f}]\subseteq F$. {\rm(iii)}
${\widetilde f}$ is the largest element in ${\widetilde F}$.
\end{lemma}

\begin{proof}
Part (i) follows from Theorem \ref{th:KFheretaSTD} (iii) and Theorem
\ref{th:tStar}. As ${\widetilde f}$ is $\eta$-invariant, $\eta\sb
{\widetilde f}(E)=E[0,{\widetilde f}]$ by Definition \ref{df:etainvar}.
Thus, if $f\in{\widetilde F}$, then by (i), $f=\eta\sb{f}f\leq\eta\sb
{\widetilde f}f\in\eta\sb{\widetilde f}(E)=E[0,{\widetilde f}]$,
whence $f\in E[0,{\widetilde f}]$, and we have ${\widetilde F}
\subseteq E[0,{\widetilde f}]$. Since ${\widetilde f}\in F$ and
$F$ is an order ideal, $E[0,{\widetilde f}]\subseteq F$, and (ii)
is proved. Part (iii) follows immediately from (ii).
\end{proof}

\begin{definition} \label{df:ftildeandu}
Henceforth, ${\widetilde f}$ denotes the largest finite invariant element
in $E$ as per Lemma \ref{lm:etasbtildeF1}.
\end{definition}

\begin{lemma} \label{lm:tildef}
The following conditions are mutually equivalent{\rm:} {\rm(i)} There exists
a faithful finite invariant element in $E$. {\rm(ii)} ${\widetilde f}$ is
faithful. {\rm(iii)} $E[0,{\widetilde f}]=E$. Moreover, if any {\rm(}hence
all{\rm)} of these conditions hold, then $E$ is an EA with unit
${\widetilde f}$ and $E=F$.
\end{lemma}

\begin{proof}
If (i) holds, then $f\leq{\widetilde f}$ and $1=\eta\sb{f}\leq\eta
\sb{{\widetilde f}}$, whence $\eta\sb{{\widetilde f}}=1$ and (ii)
holds. That (ii) $\Rightarrow$ (i) is obvious, and we have (i)
$\Leftrightarrow$ (ii). Lemma \ref{lm:etasbtildeF1} (ii) implies
that (ii) $\Leftrightarrow$ (iii). If (iii) holds, then $E[0,
{\widetilde f}]=E$, so ${\widetilde f}$ is the largest element
in $E$; also, since ${\widetilde f}\in F$ and $F$ is an order ideal,
we have $E=E[0,{\widetilde f}]\subseteq F$, whence $E=F$.
\end{proof}

We omit the straightforward proof of the following lemma.

\begin{lemma}  \label{lm:invfiniteinpiE}
If $\pi\in\Sigma\sb{\sim}(E)$, then $\pi{\widetilde f}$ is the largest
finite invariant element in the DGEA $\pi(E)$.
\end{lemma}

\section{Decomposition into Types I, II, and III} \label{sc:I/II/III}

\noindent\emph{Standing Assumptions \ref{as:DGEA} remain in force.}

\begin{definition}
(1) $E$ is of \emph{type} I iff there is a faithful simple element
($\eta$-monad) in $E$. (2) $E$ is of \emph{type} II iff there is a
faithful finite element in $E$, but there are no nonzero simple
elements in $E$. (3) $E$ is of \emph{type} III iff there are no nonzero
finite elements in $E$. (4) $E$ is of \emph{finite type} iff there is a
faithful finite invariant element in $E$. (5) $E$ is of \emph{properly
non-finite type} iff there are no nonzero invariant finite elements in
$E$. (6) $E$ is of \emph{type {\rm I}$\sb{F}$} (respectively, of
\emph{type {\rm II}$\sb{F}$}) iff it is both of type I (respectively,
of type II) and of finite type. (7) $E$ is of \emph{type} I$\sb{\neg F}$
(respectively, \emph{type} II$\sb{\neg F}$) iff it is both of type I
(respectively, type II) and of properly non-finite type.
\end{definition}

\begin{lemma} \label{lm:typesforE}
{\rm(i)} The DGEA $E$ is of type {\rm I} iff $\eta\sb{K}=1$, in which
case every element in $E$ is the orthosum of a family of simple elements
{\rm(}$\eta$-monads{\rm)}in $E$. {\rm(ii)} $E$ is of type {\rm II} iff
$\eta\sb{F}=1$ and $\eta\sb{K}=0$, in which case every element in
$E$ is the orthosum of a family of finite elements in $E$. {\rm(iii)}
$E$ is of type {\rm III} iff $\eta\sb{F}=0$.  {\rm(iv)} $E$ is of finite
type iff $\eta\sb{\widetilde f}=1$, in which case $E$ is an EA with
unit ${\widetilde f}$ and every element in $E$ is finite. {\rm(v)}
$E$ is properly non-finite iff $\eta\sb{\widetilde f}=0$, i.e., iff
${\widetilde f}=0$.\footnote{Although there are no nonzero finite invariant
elements in a properly non-finite GEA, it is possible for such a GEA to
contain nonzero finite elements.}
\end{lemma}

\begin{proof}
Parts (i), (ii), and (iii) follow from parts (iii), (iv), and (v) of Theorem
\ref{th:hstar} and Corollary \ref{co:KF}; (iv)
is a consequence of Lemma \ref{lm:tildef}; and (v) follows from Definition
\ref{df:ftildeandu}.
\end{proof}

\begin{theorem} [Cf. {\cite[Definition 13.3]{HDTD}}] \label{th:I/II/III}
Let $\pi\in\Sigma\sb{\sim}(E)$. Then{\rm:}
\begin{enumerate}
\item If the DGEA $\pi(E)$ is of type {\rm I}, of type {\rm II}, or
 of finite type, then $\pi\in\Theta\sb{\eta}(E)$.
\item $\pi(E)$ is of type {\rm I} iff $\pi\in\Theta\sb{\eta}
 (E)$ and $\pi\leq\eta\sb{K}$.
\item $\pi(E)$ is of type {\rm II} iff $\pi\in\Theta\sb{\eta}(E)$
 and $\pi\leq\eta\sb{F}\wedge(\eta\sb{K})'$.
\item $\pi(E)$ is of type {\rm III} iff $\pi\leq(\eta\sb{F})'$.
\item $\pi(E)$ is of finite type iff $\pi\in\Theta\sb{\eta}(E)$ and
 $\pi\leq\eta\sb{\widetilde f}$.
\item $\pi(E)$ is of properly non-finite type iff $\pi\wedge\eta\sb
 {\widetilde f}=0$.
\item $\pi(E)$ is of type {\rm I}$\sb{F}$ {\rm(}respectively, of type
 {\rm II}$\sb{F}${\rm)} iff $\pi\in\Theta\sb{\eta}(E)$ and $\pi\leq\eta
 \sb{K}\wedge\eta\sb{\widetilde f}$ {\rm(}respectively, iff $\pi\in
 \Theta\sb{\eta}(E)$ and $\pi\leq\eta\sb{F}\wedge(\eta\sb{K})'\wedge\eta
 \sb{{\widetilde f}}${\rm)}.
\item $\pi(E)$ is of type {\rm I}$\sb{\neg F}$ {\rm(}respectively, of type
 {\rm II}$\sb{\neg F}${\rm)} iff $\pi\in\Theta\sb{\eta}(E)$ and $\pi\leq
 \eta\sb{K}\wedge(\eta\sb{\widetilde f})'$ {\rm(}respectively, $\pi
 \in\Theta\sb{\eta}(E)$ and $\pi\leq\eta\sb{F}\wedge(\eta\sb{K})'\wedge
 (\eta\sb{\widetilde f})'${\rm)}.
\end{enumerate}
\end{theorem}

\begin{proof}
According to Lemmas \ref{lm:KcappiE} and \ref{lm:invfiniteinpiE}, $K\cap
\pi(E)$ is the set of simple elements in $\pi(E)$, $F\cap\pi(E)$ is
the set of finite elements in $\pi(E)$, and $\pi{\widetilde f}$ is
the largest invariant finite element in $\pi(E)$. Thus (i) is a
consequence of Lemma \ref{lm:pfaithful}; (ii) follows from (i) and
Theorem \ref{th:KFfaithinpi}; (iii) follows from (i), Theorem
\ref{th:KFfaithinpi} and Lemma \ref{lm:HcappiE=0}; and (iv) is
a consequence of Lemma \ref{lm:HcappiE=0}.

To prove (v), suppose that $\pi(E)$ is of finite type, i.e., there
exists $f\in\pi(E)$ such that $f$ is finite, invariant, and faithful
in the DGEA $\pi(E)$. Then $f\in F$ by Lemma \ref{lm:KcappiE} (ii)
and $f\in\Gamma\sb{\sim}(E)$ by Lemma \ref{th:piEasDGEA} (v), whence
$f\leq{\widetilde f}$ by Lemma \ref{lm:etasbtildeF1} (iii). Thus, as
$f$ is faithful in $\pi(E)$, we infer from Lemma \ref{lm:pfaithful}
that $\pi\leq\eta\sb{f}\leq\eta\sb{\widetilde f}$. Also, $\pi
\in\Theta\sb{\eta}(E)$ by (i). Conversely, suppose that $\pi\in\Theta
\sb{\eta}(E)$ and $\pi\leq\eta\sb{\widetilde f}$. As $\pi\in\Theta\sb
{\eta}(E)$, there exists $e\in E$ such that $\pi=\eta\sb{e}$, and we
have $\pi=\pi\wedge\eta\sb{\widetilde f}=\eta\sb{e}\wedge\eta
\sb{\widetilde f}=\eta\sb{\eta\sb{e}{\widetilde f}}=\eta\sb
{\pi{\widetilde f}}$. According to Lemma \ref{lm:invfiniteinpiE},
$\pi{\widetilde f}$ is a finite invariant element in the DGEA
$\pi(E)$, and $\pi=\eta\sb{\pi{\widetilde f}}$ implies that
$\pi{\widetilde f}$ is faithful in $\pi(E)$ by Lemma \ref
{lm:pfaithful}: hence $\pi(E)$ is of finite type, and (v) is
proved.

As per Lemma \ref{lm:invfiniteinpiE}, $\pi(E)$ is of properly
non-finite type iff $\pi{\widetilde f}=0$, i.e., iff $\pi\wedge
\eta\sb{\widetilde f}=0$ by Theorem \ref{th:SigmasimProps} (v).
This proves part (vi), and parts (vii) and (viii) follow from .
parts (i)--(vi).
\end{proof}

We are now in a position to state and prove the \emph{fundamental
type {\rm I/II/III}-decomposition theorem for DGEAs.}

\begin{theorem} \label{th:DecompositionI/II/III}
There are pairwise disjoint mappings $\pi\sb{\rm I},\,\pi\sb{\rm II},
\,\pi\sb{\rm III}\in\Sigma{\sb{\sim}(E)}$ such that $\pi\sb{\rm I}
\vee\pi\sb{\rm II}\vee\pi\sb{\rm III}=1$ and $E$ decomposes as a
direct sum
\[
E=\pi\sb{\rm I}(E)\oplus \pi\sb{\rm II}(E)\oplus\pi\sb{\rm III}(E)
\]
where $\pi\sb{\rm I}(E),\, \pi\sb{\rm II}(E)$, and $\pi\sb{\rm III}(E)$
are hereditary direct summands of types {\rm I}, {\rm II}, and {\rm III},
respectively. These mappings are unique, in fact
\[
\pi\sb{\rm I}=\eta\sb{K}\in\Theta\sb{\eta}(E),\ \pi\sb{\rm II}=
 \eta\sb{F}\wedge(\eta\sb{K})'\in\Theta\sb{\eta}(E),\text{\ and\ }
 \pi\sb{\rm III}=(\eta\sb{F})'.
\]
Moreover, there are pairwise disjoint mappings $\pi\sb{{\rm I}\sb{F}},
\,\pi\sb{{\rm I}\sb{\neg F}},\,\pi\sb{{\rm II}\sb{F}},\,\pi\sb{{\rm II}
\sb{\neg F}}\in\Sigma\sb{\sim}(E)$ such that $\pi\sb{{\rm I}\sb{F}}
\vee\pi\sb{{\rm I}\sb{\neg F}}=\pi\sb{\rm I}$ and $\pi\sb{{\rm II}\sb{F}}
\vee\pi\sb{{\rm II}\sb{\neg F}}=\pi\sb{\rm II}$; these mappings induce
further decompositions
\[
\pi\sb{\rm I}(E)=\pi\sb{{\rm I}\sb{F}}(E)\oplus\pi\sb{{\rm I}
 \sb{\neg F}}(E)\text{\ and\ }\pi\sb{\rm II}(E)=\pi\sb{{\rm II}
 \sb{F}}(E)\oplus\pi\sb{{\rm II}\sb{\neg F}}(E)
\]
where $\pi\sb{{\rm I}\sb{F}}(E),\,\pi\sb{{\rm I}\sb{\neg F}}(E),\,
\pi\sb{{\rm II}\sb{F}}(E)$ and $\pi\sb{{\rm II}\sb{\neg F}}(E)$
are hereditary direct summands of types {\rm I}$\sb{F}$,
{\rm I}$\sb{\neg F}$, {\rm II}$\sb{F}$, and {\rm II}$\sb{\neg F}$,
respectively. These mappings are also unique, in fact
\[
\pi\sb{{\rm I}\sb{F}}=\eta\sb{K}\wedge\eta\sb{\widetilde f}\in
 \Theta\sb{\eta}(E),\ \pi\sb{{\rm I}\sb{\neg F}}=\eta\sb{K}
 \wedge(\eta\sb{\widetilde f})'\in\Theta\sb{\eta}(E),
\]
\[
\pi\sb{{\rm II}\sb{F}}=\eta\sb{F}\wedge(\eta\sb{K})'\wedge\eta
 \sb{\widetilde f}\in\Theta\sb{\eta}(E),\text{\ and\ }\pi\sb
 {{\rm II}\sb{\neg F}}=\eta\sb{F}\wedge(\eta\sb{K})'\wedge
 (\eta\sb{\widetilde f})'\in\Theta\sb{\eta}(E).
\]
Furthermore, the hereditary direct summands $\pi\sb{{\rm I}
\sb{F}}(E)$ and $\pi\sb{{\rm II}\sb{F}}(E)$ are EAs with
units $\eta\sb{K}{\widetilde f}$ and $(\eta\sb{K})'{\widetilde f}$,
respectively.
\end{theorem}

\begin{proof}
Put $\pi\sb{\rm I}:=\eta\sb{K}$, $\pi\sb{\rm II}:=\eta\sb{F}\wedge
(\eta\sb{K})'$, $\pi\sb{\rm III}:=(\eta\sb{F})'$, $\pi\sb{{\rm I}
\sb{F}}:=\eta\sb{K}\wedge\eta\sb{\widetilde f}$, $\pi\sb{{\rm I}
\sb{\neg F}}=\eta\sb{K}\wedge(\eta\sb{\widetilde f})'$, $\pi\sb
{{\rm II}\sb{F}}=\eta\sb{F}\wedge(\eta\sb{K})'\wedge\eta\sb
{\widetilde f}$, and $\pi\sb{{\rm II}\sb{\neg F}}=\eta\sb{F}
\wedge(\eta\sb{K})'\wedge(\eta\sb{\widetilde f})'$. As $\eta\sb{K}
\leq\eta\sb{F}$, these mappings are pairwise disjoint, and
it is clear that $\pi\sb{\rm I}\vee\pi\sb{\rm II}\vee\pi\sb{\rm III}
=1$, $\pi\sb{{\rm I}\sb{F}}\vee\pi\sb{{\rm I}\sb{\neg F}}=\pi
\sb{\rm I}$, and $\pi\sb{{\rm II}\sb{F}}\vee\pi\sb{{\rm II}
\sb{\neg F}}=\pi\sb{\rm II}$. By Theorem \ref{th:I/II/III}, the
various hereditary direct summands corresponding to these mappings
are of the types indicated in the statement of the theorem. By
Corollary \ref{co:KF} $\pi\sb{\rm I}=\eta\sb{K}\in\Theta\sb{\eta}(E)$
and $\eta\sb{F}\in\Theta\sb{\eta}(E)$. Since $\Theta\sb{\eta}(E)$ is
a hull-determining subset of $\GEX(E)$, it follows that $\pi\sb{\rm II}
=\eta\sb{F}\wedge(\eta\sb{K})'\in\Theta\sb{\eta}(E)$, and similar
arguments show that $\pi\sb{{\rm I}\sb{F}}, \,\pi\sb{{\rm I}\sb
{\neg F}},\,\pi\sb{{\rm II}\sb{F}},\,\pi\sb{{\rm II}\sb{\neg F}}\in
\Theta\sb{\eta}(E)$. By Lemmas \ref{lm:tildef} and \ref
{lm:invfiniteinpiE}, $\pi\sb{{\rm I}\sb{F}}(E)$ and $\pi\sb{{\rm II}
\sb{F}}(E)$ are EAs with units $\pi\sb{{\rm I}\sb{F}}{\widetilde f}
=\eta\sb{K}{\widetilde f}\wedge\eta\sb{\widetilde f}{\widetilde f}
=\eta\sb{K}{\widetilde f}\wedge{\widetilde f}=\eta\sb{K}{\widetilde f}$
and $\pi\sb{{\rm II}\sb{F}}{\widetilde f}=\eta\sb{F}{\widetilde f}
\wedge(\eta\sb{K})'{\widetilde f}\wedge\eta\sb{\widetilde f}
{\widetilde f}=\eta\sb{F}{\widetilde f}\wedge(\eta\sb{K})'
{\widetilde f}$. But ${\widetilde f}=\eta\sb{\widetilde f}
{\widetilde f}\leq\bigvee\sb{f\in F}\eta\sb{f}{\widetilde f}=
\eta\sb{F}{\widetilde f}\leq{\widetilde f}$, so $\eta\sb{F}
{\widetilde f}={\widetilde f}$, and therefore $\pi\sb{{\rm II}
\sb{F}}{\widetilde f}=(\eta\sb{K})'{\widetilde f}$.

To prove the uniqueness, suppose $\pi\sb{\rm I},\,\pi\sb{\rm II},
\,\pi\sb{\rm III}$ are mappings in $\Sigma{\sb{\sim}(E)}$ with
$\pi\sb{\rm I}\vee\pi\sb{\rm II}\vee\pi\sb{\rm III}=1$ and that
$\pi\sb{\rm I}(E),\,\pi\sb{\rm II}(E)$ and $\pi\sb{\rm III}(E)$
are of types I, II, and III, respectively. Then by Theorem \ref
{th:I/II/III}, $\pi\sb{\rm I}\leq\eta\sb{K}$, $\pi\sb{\rm II}
\leq\eta\sb{F}\wedge(\eta\sb{K})'$, and $\pi\sb{\rm III}\leq
(\eta\sb{F})'$. As $\eta\sb{K}\leq\eta\sb{F}$, it follows that
$\eta\sb{K}$, $\eta\sb{F}\wedge(\eta\sb{K})'$, and $(\eta\sb{F})'$
are pairwise disjoint with $\eta\sb{K}\vee\eta\sb{F}\wedge
(\eta\sb{K})'\vee(\eta\sb{F})'=1$, and by a simple boolean argument,
it follows that $\pi\sb{\rm I}=\eta\sb{K}$, $\pi\sb{\rm II}=
\eta\sb{F}\wedge(\eta\sb{K})'$, and $\pi\sb{\rm III}=(\eta\sb{F})'$.
Similarly, one proves the uniqueness of $\pi\sb{{\rm I}\sb{F}},
\,\pi\sb{{\rm I}\sb{\neg F}},\,\pi\sb{{\rm II}\sb{F}}$ and $\pi
\sb{{\rm II}\sb{\neg F}}$.
\end{proof}


\begin{thebibliography}{99}
%
\bibitem{DvPuTrends} Dvure\v{c}enskij, A. and Pulmannov\'{a}, S., \emph{New
Trends in Quantum Structures}, Kluwer, Dordrecht, 2000.

\bibitem{FandB} Foulis, D.J. and Bennett, M.K., Effect algebras and
unsharp quantum logics, \emph{Found. Phys.} {\bf 24} (1994) 1331--1352.

\bibitem{COEA} Foulis, D.J. and Pulmannov\'{a}, S., Centrally orthocomplete
effect algebras, \emph{Algebra Univ.} {\bf 64} (2010) 283--307,
DOI 10.1007/s00012-010-0100-5.

\bibitem{HandD} Foulis, D.J. and Pulmannov\'{a}, S., Hull mappings and
dimension effect algebras, \emph{Math. Slovaca} {\bf 61}, no. 3 (2011)
485--522, DOI 10.2478/s12175-011-0025-2.

\bibitem{ExoCen} Foulis, D.J. and Pulmannov\'{a}, S., The exocenter of a
generalized effect algebra, \emph{Rep. Math. Phys.} {\bf 68} no. 3 (2011)
347--371.

\bibitem{CenGEA} Foulis, D.J. and Pulmannov\'{a}, S., The center of a
generalized effect algebra, submitted.

\bibitem{HDTD} Foulis, D.J. and Pulmannov\'{a}, S., Hull determination and
type decomposition for a generalized effect algebra, submitted.

\bibitem{GFP} Greechie, R.J., Foulis, D.J, and Pulmannov\'{a}, S., The
center of an effect algebra, \emph{Order} {\bf 12} (1995) 91--106.

\bibitem{Good} Goodearl, K.R., \emph{Partially Ordered Abelian Groups with
Interpolation}, Math. Surveys and Monographs No. 20, Amer. Math. Soc.,
Providence, RI, 1986; ISBN 0-8218-1520-2.

\bibitem{GWDim} Goodearl, K.R. and Wehrung, F., \emph{The Complete
 Dimension Theory of Partially Ordered Systems with Equivalence and
 Orthogonality}, Memoirs of AMS 831, AMS, Providence, Rhode Island, 2005.

\bibitem{HedPu} Hedl\'{i}kov\'{a}, J. and Pulmannov\'{a}, S.J., Generalized
difference posets and orthoalgebras, \emph{Acta Math. Univ. Comenianae}
{\bf 45} (1996) 247--279.

\bibitem{Je00} Jen\v{c}a, G., Subcentral ideals in generalized effect algebras,
\emph{Int. J. Theor. Phys.} {\bf 39} (2000) 745--755.

\bibitem{Je02} Jen\v{c}a, G., A Cantor-Bernstein type theorem for effect
algebras, \emph{Algebra Univers.} {\bf 48} (2002) 399--411.

\bibitem{JePu02} Jen\v{c}a, G. and Pulmannov\'{a}, S., Quotients of partial
abelian monoids and the Riesz decomposition property, \emph{Algebra Univers.}
{\bf 47} (2002), 443-477.

\bibitem{Kalm86} Kalmbach, G., \emph{Measures and Hilbert lattices}, World
Scientific Publishing Co., Singapore, 1986. ISBN: 9971-50-009-4.

\bibitem{K} Kalinin, V.V., Orthomodular partially ordered sets with dimension
(Russian), \emph{Algebra i logika} {\bf 15}, no. 5 (1987) 535-557.

\bibitem{KalmZdenka} Kalmbach, G. and Rie\v{c}anov\'{a}, Z., An axiomatization
for abelian relative inverses, \emph{Dem. Math.} {\bf 27} (1994) 535--537.

\bibitem{Loom} Loomis, L.H., \emph{\it The lattice theoretic background of
the dimension theory of operator algebras}, Memoirs  of  the AMS {\bf 18} (1955).

\bibitem{Mae} Maeda, S., {\it Dimension function on certain general lattices},
J. Sci. Hiroshima Univ. {\bf A19} (1955), 211-237.

\bibitem{Pola} Polakovi\v{c}, M., Generalized effect algebras of positive
operators defined on Hilbert spaces, \emph{Rep. Math. Phys.} {\bf 68} (2011)
241--250.

\bibitem{PZOps} Polakovi\v{c}, M. and Rie\v{c}anov\'{a}, Z., Generalized
effect algebras of positive operators densely defined on Hilbert spaces,
\emph{Int. J. of Theor. Phys.} {\bf 50} (2010) 1167--1174, DOI
10.1007/s10773-010-0458-3 (2010).

\bibitem{PuCong} Pulmannov\'{a}, S.,  Congruences in partial Abelian semigroups,
\emph{Algebra Univers.} {\bf 37} (1997) 119--140.

\bibitem{PVQObs} Pulmannov\'a, S., Vincekov\'a, E., Remarks on the order for
quantum observables, \emph{Math. Slovaca} {\bf 57}, no. 6 (2007) 589--600,
DOI 10.2478/s12175-007-0048-x.

\bibitem{PVRiesz} Pulmannov\'a, S., Vincekov\'a, E., Riesz ideals in generalized
effect algebras and in their unitizations, \emph{Algebra. Univers.} {\rm 57} (2007)
393--417, DOI 10.1007/x00012-007-2043-z.

\bibitem{RZP} Rie\v{c}annov\'{a}, Z., Zajac, M., and Pulmannov\'{a}, S.,
Effect algebras of positive linear operators densely defined on
Hilbert spaces, \emph{Rep. Math. Phys.} {\bf 68} (2011)
261--270.

\bibitem{Ram} Ramsay, A.,  Dimension theory in complete weakly modular
orthocomplemented lattices, \emph{Trans. Amer. Math. Soc.} {\bf 116}
(1965), 9-31.

\bibitem{Sh} Sherstnev, A.N., On boolean logics (Russian), \emph{U\v{c}ebnyje
zapisky Kazanskogo universiteta}, No. 2 (1968) 48-62.

\bibitem{Tkad} Tkadlec, J., Atomistic and orthoatomistic effect algebras,
\emph{J. Math. Phys.} {\bf 49} no. 5 (2008) 5 pp.

\bibitem{ZR98} Rie\v{c}anov\'{a}, Z., Compatibility and central elements in
effect algebras. Fuzzy sets, Part I (Liptovsk\'{y} J\'{a}n, 1998)
\emph{Tatra Mt. Math. Publ.} {\bf 16} (1999) 151--158.

\bibitem{ZR01} Rie\v{c}anov\'{a}, Z., Orthogonal sets in effect algebras,
\emph{Dem. Math.} {\bf 34}, no. 3 (2001) 525--532.

\bibitem{ZR09} Rie\v{c}anov\'{a}, Z., Pseudocomplemented lattice effect
algebras and existence of states, \emph{Inform. Sci.} {\bf 179}, no. 5
(2009) 529--534.

\bibitem{Zdenka99} Rie\v{c}anov\'{a}, Z., Subalgebras, intervals, and central
elements of generalized effect algebras, \emph{Int. J. Theor. Phys.} {\bf 38},
no. 12 (1999) 3209--3220.

\bibitem{AlexPAS} Wilce, A., Perspectivity and congruence in partial abelian semigroups,
\emph{Math. Slovaca} {\bf 48} (1998) 117--135.


\end{thebibliography}
\end{document}